\documentclass[11pt]{article}
\usepackage[bottom=1in,left=1in, right=1in, top=1in]{geometry}
\usepackage{microtype}
\usepackage{graphicx}
\usepackage{subfigure}
\usepackage{booktabs}
\usepackage{hyperref}
 
\usepackage{amsmath}
\usepackage{amssymb}
\usepackage{mathtools}
\usepackage{amsthm}
\usepackage{thmtools, thm-restate}
\usepackage{notation}
\usepackage{enumitem}

\newtheorem{claim}{Claim}

\usepackage[capitalize,noabbrev]{cleveref}

\usepackage[textsize=tiny]{todonotes}

\usepackage{multirow}

\usepackage{times}
\usepackage{fancyhdr}
\usepackage{xcolor}
\usepackage{algorithm}
\usepackage{algorithmic}
\usepackage{natbib}
\usepackage{eso-pic}
\usepackage{forloop}
\usepackage{url}
\usepackage{subfigure}

\usepackage{notation}

\renewcommand{\cite}[1]{\citep{#1}}

\usepackage{authblk}

\title{Edge-Colored Clustering in Hypergraphs: Beyond Minimizing Unsatisfied Edges}
\date{}
\author[1]{Alex Crane}
\author[2]{Thomas Stanley}
\author[1]{Blair D. Sullivan}
\author[2]{Nate Veldt}

\affil[1]{University of Utah}
\affil[2]{Texas A\&M University}

\begin{document}

\maketitle

\begin{abstract}
    We consider a framework for clustering edge-colored hypergraphs, where the goal is to cluster (equivalently, to \emph{color}) objects based on the primary {type} of multiway interactions they participate in. One well-studied objective is to color nodes to minimize the number of \textit{unsatisfied} hyperedges -- those containing one or more nodes whose color does not match the hyperedge color. We motivate and present advances for several directions that extend beyond this minimization problem. We first provide new algorithms for maximizing \textit{satisfied} edges, which is the same at optimality but is much more challenging to approximate, with all prior work restricted to graphs. We develop the first approximation algorithm for hypergraphs, and then refine it to improve the best-known approximation factor for graphs. We then introduce new objective functions that incorporate notions of balance and fairness, and provide new hardness results, approximations, and fixed-parameter tractability results.
\end{abstract}

\section{Introduction}
\label{sec:intro}
Edge-colored clustering (ECC) is an optimization framework for clustering datasets characterized by {categorical} relationships among data points. The problem is formally encoded as an edge-colored hypergraph (Figure~\ref{fig:eccbalance}), where each edge represents an interaction between data objects (the nodes) and the color of the edge indicates the \emph{type} or \emph{category} of that interaction. The goal is to assign colors to nodes in such a way that edges of a color tend to include nodes of that color, by minimizing or maximizing some objective function relating edge colors and node colors.
ECC algorithms have been applied to various clustering tasks where cluster labels naturally match with interaction types. For example, if nodes are researchers, edges are author lists for publications, and colors indicate publication field (computer science, biology, etc.), then ECC provides a framework for inferring researchers' fields based on publications. ECC has also been used for temporal hypergraph clustering~\cite{amburg2020clustering}, where edge colors encode time windows in which interactions occur. ECC then clusters nodes into time windows in which they are especially active. Variants of ECC have also been used for team formation~\cite{amburg2022diverse}, in which case nodes are people, edges represent team tasks, and colors indicate task type. In this setting, ECC corresponds to assigning tasks based on prior team experiences.
\begin{figure}[t]
	\centering
	\includegraphics[width = .6\linewidth]{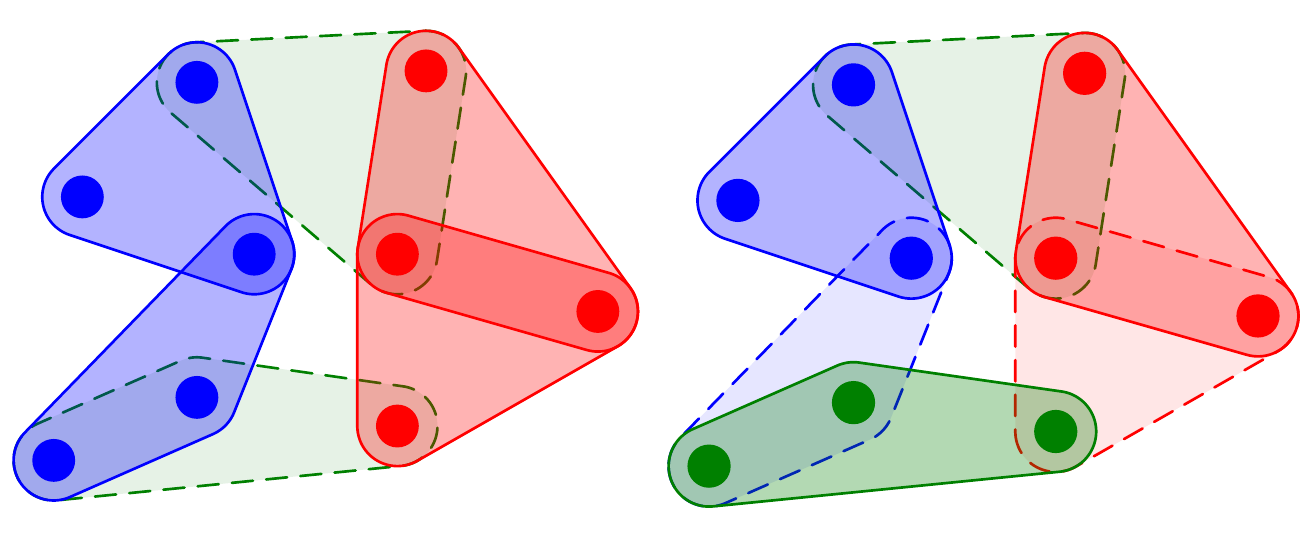}
	\vspace{-5pt}
	\caption{The node coloring on the left satisfies 4 edges (2 blue, 2 red, 0 green). The coloring on the right only satisfies 3, but each color has a satisfied edge.}
	\label{fig:eccbalance}
\end{figure}

\textbf{Related work and research gaps.}
Edge-colored clustering has been well-studied in the machine learning and data mining literature from the perspective of approximation algorithms~\cite{amburg2020clustering,amburg2022diverse,veldt2023optimal,crane2024overlapping,angel2016clustering} and fixed-parameter tractability results~\cite{kellerhals2023parameterized,cai2018alternating,crane2024overlapping}. It is also closely related to chromatic correlation clustering~\cite{bonchi2012chromatic,bonchi2015chromatic,anava2015improved,klodt2021color,xiu2022chromatic}, which is an edge-colored variant of correlation clustering~\cite{bansal2004correlation}.
Several variants of ECC have been encoded using different combinatorial objective functions for assigning colors to nodes. The earliest and arguably the most natural is \maxecc{}~\cite{angel2016clustering}, which seeks to maximize the number of \textit{satisfied} edges---edges in which all the nodes match the color of the edge. More recent attention has been paid to \minecc{}~\cite{amburg2020clustering}, where the goal is to minimize the number of \textit{unsatisfied} edges. The latter is the same as \maxecc{} at optimality but differs in terms of approximations and fixed-parameter tractability results.
The approximability of \minecc{} is currently well-understood: a recent ICML paper designed improved approximation guarantees that are tight with respect to linear programming integrality gaps and nearly tight with respect to approximation hardness bounds~\cite{veldt2023optimal}.
However, existing approximations for \maxecc{} apply only to graph inputs (whereas nearly all algorithms for \minecc{} apply to hypergraphs), and there is a much larger gap between the best approximation guarantees and known lower bounds.

Another limitation of prior work is that nearly all objective functions for ECC (including \minecc{} and \maxecc{}) only consider the total number of edges that are (un)satisfied, even if it means certain colors are disproportionately (un)satisfied. Figure~\ref{fig:eccbalance} provides a small example where the optimal  \minecc{} solution satisfies four edges but only from two colors. Meanwhile, another color assignment satisfies only three edges but includes a satisfied edge of each color. The latter choice is natural for applications where one may wish to incorporate some notion of balance or fairness in edge satisfaction. For example, if coloring nodes means assigning individuals to tasks for future team interactions, then we would like to avoid situations where one type of task is assigned no workers. Balanced and fair objectives have been studied and applied for many other recent clustering frameworks, but have yet to be explored in the context of edge-colored clustering. For prior work on fair clustering, we refer to the seminal work of~\citet{chierichetti2017fair} and a recent survey by~\citet{caton2024fairness}.

Motivated by the above, we present new contributions to ECC along two frontiers: improved  \maxecc{} algorithms, and new frameworks for balanced and fair ECC variants.

\textbf{Our contributions for \maxecc{}}.
We present the first approximation algorithm for hypergraph \maxecc{}, which has an approximation factor of
$(2/e)^r(r+1)^{-1}$
where $r$ is the maximum hyperedge size. The approximation factor goes to zero as $r$ increases, but this is expected since a prior hardness result rules out the possibility of constant-factor approximations that hold for arbitrarily large $r$~\cite{veldt2023optimal}. Our result shows that non-trivial constant-factor approximations can be obtained when $r$ is a constant. We use insights from our hypergraph algorithm to obtain a new best approximation factor of $154/405 \approx 0.38$ for the graph version of the problem, improving on the previous best 0.3622-approximation shown five years ago~\cite{ageev20200}. While the increase in approximation factor appears small at face value, this improves on a long sequence of papers on \maxecc{} for the graph case~\cite{angel2016clustering,ageev2015improved,alhamdan2019approximability,ageev20200}. Obtaining even a minor increase in the approximation factor is highly non-trivial and constitutes the most technical result of our paper.

\textbf{Our contributions for balanced and fair ECC variants.}
Our first contribution towards balanced and fair variants of ECC is to introduce the generalized \pmeanECC{} objective, which uses a parameter $p$ to control balance in unsatisfied edges, and captures \minecc{} as a special case when $p = 1$. For $p = \infty$, the objective corresponds to a new problem we call \cfminECC{}, where the goal is to minimize the maximum number of edges of any color that are unsatisfied. We prove \cfminECC{} is NP-hard even for $k = 2$ colors, even though standard \minecc{} is polynomial-time solvable for $k = 2$. We then show how to obtain a 2-approximation by rounding a linear program (LP), which matches an LP integrality gap lower bound. More generally, we obtain a 2-approximation for \pmeanECC{} for every $p \geq 1$ by rounding a convex relaxation, and give a $2^{1/p}$-approximation for every $p \in (0,1)$. Our results also include several parameterized complexity and hardness results for \cfminECC{} and a related maximization variant where the goal is to maximize the minimum number of edges of any one color that are satisfied.

We also consider scenarios in which the ``importance'' of the edge colors is unequal. In particular, we
introduce \pcECC{}, where we are given a special color $c_1$ which can be thought of as encoding a protected interaction type. The goal is to color nodes to minimize the number of unsatisfied edges subject to a strict upper bound on the number of unsatisfied edges of color $c_1$. We give bicriteria approximation algorithms based on rounding a linear programming relaxation, meaning that our algorithm is allowed to violate the upper bound for the protected color by a bounded amount, in order to find a solution that has a bounded number of unsatisfied edges relative to the optimal solution.
Finally, although the contributions of our paper are primarily theoretical, we also implement our algorithms for \cfminECC{} and \pcECC{} on a suite of benchmark ECC datasets. Our empirical results show that (1) our algorithms outperform their theoretical guarantees in practice, and (2) our balanced and fair objectives still achieve good results with respect to the standard \minecc{} objective while better incorporating notions of balance and fairness.

\section{Improved \maxecc{} Algorithms}
\textbf{Preliminaries.} For a positive integer $n$, let $[n] = \{1,2, \hdots n\}$.
We use bold lowercase letters to denote vectors, and indicate the $i$th of entry of a vector $\textbf{x} \in \mathbb{R}^n$ by $x_i$. For a set $S$ and a positive integer $t$, let $\binom{S}{t}$ denote all subsets of $S$ of size $t$.
An instance of edge-colored clustering is given by a hypergraph $H = (V,E, \ell)$ where $V$ is a node set, $E$ is a set of hyperedges (usually just called \emph{edges}), and $\ell \colon E \rightarrow [k]$ is a mapping from edges to a color set $[k] = \{1,2, \hdots, k\}$.  For $e \in E$, we let $\omega_e \geq 0 $ denote a nonnegative weight associated with $e$, which equals 1 for all edges in the unweighted version of the problem. For a color $c \in [k]$, let $E_c \subseteq E$ denote the edges of color $c$. We use $r$ to denote the \emph{rank} of $H$, i.e., the maximum hyperedge size.

The goal of ECC is to construct a map $\lambda \colon V \rightarrow [k]$ that associates each node with a color, in order to optimize some function on edge \emph{satisfaction}. Edge $e \in E$ is \emph{satisfied} if $\ell(e) = \lambda(v)$ for each $v \in e$, and is otherwise \emph{unsatisfied}. The two most common objectives are maximizing the number of satisfied edges (\maxecc{}) or minimizing the number of unsatisfied edges (\minecc{}).

\textbf{Challenges in approximating \maxecc{}.} Although \minecc{} and \maxecc{} are equivalent at optimality, the latter is far more challenging to approximate. Due to an approximation-preserving reduction from \textsc{Independent Set}, it is NP-hard
to approximate \maxecc{} in hypergraphs of unbounded rank $r$ to within a factor $|E|^{1-\varepsilon}$~\cite{veldt2023optimal,zuckerman2006linear}. There also are simple instances (e.g., a triangle with 3 colors) where a simple $2$-approximation of~\citet{amburg2020clustering} for \minecc{} (round variables of an LP relaxation to 0 if they are strictly below $1/2$, otherwise round to 1) fails to satisfy \emph{any} edges.
These challenges do not rule out the possibility of approximating \maxecc{} when $r$ is constant. Indeed, there are many approximations for graph \maxecc{} ($r=2$), but these require lengthy proofs, rely fundamentally on the assumption that the input is a graph, and do not easily extend even to the $r = 3$ case. Here we provide a generalized approach that gives the first approximation guarantees for hypergraph \maxecc{}, when $r$ is constant. Our approach also provides a simplified way to approximate graph \maxecc{}; we design and analyze a refined algorithm that achieves a new best approximation factor for graph \maxecc{}, improving on a long line of previous algorithms~\cite{angel2016clustering,alhamdan2019approximability,ageev2015improved,ageev20200}.

\subsection{Technical preliminaries for LP rounding algorithms}
\maxecc{} can be cast as a binary linear program (BLP):
\begin{align}
	\label{eq:maxecc}
	\begin{aligned}
		\text{max} \quad  & \textstyle \sum_{e \in E} \omega_e z_e \\
		\text{s.t.} \quad & \forall v \in V:                       \\
		                  & \forall c \in [k], e \in E_c:          \\
		                  & x_v^c, z_e \in \{0, 1\}
	\end{aligned}
	\begin{aligned}
		      &                                      \\
		      & \textstyle \sum_{c=1}^k x_v^c = 1    \\
		\quad & x_v^c \geq z_e \quad \forall v \in e \\
		      & \forall c \in [k], v \in V, e \in E.
	\end{aligned}
\end{align}
Setting $x_{u}^c$ to 1 indicates that node $u$ is given color $c$ (i.e., $\lambda(u) = c$). We use $\textbf{x}_u = \begin{bmatrix} x_u^1 & x_u^2 & \cdots & x_u^k \end{bmatrix}$ to denote the vector of variables for node $u$. For edge $e \in E$, the constraints are designed in such a way that $z_e = 1$ if and only if $e$ is satisfied. A binary LP for \minecc{} can be obtained by changing the objective function to $\min \sum_{e \in E} \omega_e (1-z_e)$.

The LP relaxation for \maxecc{} can be obtained by relaxing the binary constraints in Binary Linear  Program~\eqref{eq:maxecc} to linear constraints $0 \leq x_v^c \leq 1$ and $0 \leq z_e \leq 1$. Solving this LP gives a fractional node-color assignment $x_v^c \in [0,1]$. The closer $x_v^c$ is to 1, the stronger this indicates node $v$ should be given color $c$. Our task is to round variables to assign one color to each node in a way that satisfies provable approximation guarantees.

Our algorithms (Algorithms~\ref{alg:hyper_maxecc} and~\ref{alg:graph_max_ecc}) are randomized. Both use the same random process to identify colors that a node ``wants''. For each $c \in [k]$ we {independently} generate a uniform random \emph{color threshold} $\alpha_c \in [0,1]$. If $x_u^c > \alpha_c$, we say that \emph{node $u$ wants color $c$}, or equivalently that color $c$ wants node $u$. Because a node may want more than one color, we use a random process to choose one color to assign, informed by the colors the node wants. Our approximation proofs rely on (often subtle) arguments about events that are independent from each other. We begin by presenting several useful observations that will aid in proving our results.

Our first observation is that if we can bound the expected cost of every edge in terms of the LP upper bound, it provides an overall expected approximation guarantee.
\begin{observation}
	\label{obs:prob}
	Let $\mathcal{A}$ be a randomized ECC algorithm and $p \in [0,1]$ be a fixed constant. If for each $e \in E$ we have $\prob[\text{$e$ is satisfied by  $\mathcal{A}$}] \geq p z_e$, then $\mathcal{A}$ is a $p$-approximation.
\end{observation}
Let $X_u^c$ denote the event that $u$ wants $c$, and $Z_e$ be the event that every node in edge $e \in E$ wants color $c = \ell(e)$.
\begin{observation}
	\label{obs:xuc}
	For each node $v \in V$, the events $\{X_v^c\}_{c \in [k]}$ are independent, and $\prob[X_v^c] = \prob[\alpha_c < x_v^c] = x_v^c \leq 1$.
\end{observation}
\begin{observation}
	\label{obs:edgewants}
	$\prob[Z_e] = \prob[\cap_{v \in e} X_v^c] = \min_{v \in e} x_v^c = z_e$.
\end{observation}

For an edge $e$ and color $i \neq \ell(e)$, we frequently wish to quantify the possibility that some node in $e$ wants color $i$, as this opens up the possibility that $e$ will be unsatisfied because some node $v \in e$ is given color $i$. Towards this goal, for each $e \in E$ and color $i \in [k]$ we identify one node in $e$ that has the highest likelihood of wanting $i$. Formally, we identify some \emph{representative} node $v \in e$ satisfying $x_v^i \geq x_u^i$ for every $u \in e$ (breaking ties arbitrarily if multiple nodes satisfy this), and we define $\sigma_e(i) = v$.
The definition of $\sigma_e(i)$ implies the following useful observation:
\begin{observation}
	\label{obs:cv}
	Color $i$ does not want \textbf{any} nodes in $e$ $\iff$ color $i$ does not want $v = \sigma_e(i)$.
\end{observation}
Variations of Observations~\ref{obs:prob}-\ref{obs:edgewants} have often been used to prove guarantees for previous randomized ECC algorithms. However, Observation~\ref{obs:cv} is new and is key to our analysis.

\subsection{Hypergraph MaxECC Algorithm}
In addition to color thresholds $\{\alpha_c\}$, our hypergraph \maxecc{} algorithm (Algorithm~\ref{alg:hyper_maxecc}) generates a uniform random permutation  $\pi$ to define priorities for colors. A node is then assigned to the highest priority color it wants.
\begin{algorithm}[t]
	\caption{Approximation alg. for hypergraph \maxecc{}}
	\label{alg:hyper_maxecc}
	\begin{algorithmic}
		\STATE Obtain optimal variables $\{z_e; x_v^c\}$ for the LP relaxation
		\STATE $\pi \leftarrow \text{uniform random ordering of colors } [k]$
		\STATE For $c \in [k]$, $\alpha_c \leftarrow $ uniform random threshold in $[0,1]$
		\FOR{$v \in V$}
		\STATE $\mathcal{W} = \{c \in [k] \colon \alpha_c < x_v^c\}$
		\IF{$|\mathcal{W}| > 0$}
		\STATE $\lambda(v) \leftarrow \argmax_{c \in \mathcal{W}} \pi(c)$
		\ELSE
		\STATE $\lambda(v) \leftarrow $ arbitrary color
		\ENDIF
		\ENDFOR
	\end{algorithmic}
\end{algorithm}

\begin{theorem}
	\label{thm:hypermaxecc}
	Algorithm \ref{alg:hyper_maxecc} is a $\left(\frac{1}{r+1} \left(\frac{2}{e}\right)^r\right)$-approximation algorithm for \maxecc{} in hypergraphs with rank $r$.
\end{theorem}
\begin{proof}
	Fix an arbitrary edge $e \in E$ and let $c = \ell(e)$. Let $T_e$ denote the event that $e$ is satisfied.
	By Observation~\ref{obs:prob}, it suffices to show $\prob[T_e] \geq \frac{z_e}{r+1} \left(\frac2e\right)^r$.

	Let $C = [k]\setminus \{c\}$.
	To partition the color set $C$, for each $v \in e$ we define
	$C_v = \{i \in C \colon \sigma_e(i) = v\}$.
	Recall that $\sigma_e(i)$ identifies a node $v \in e$ satisfying $x_v^i = \max_{u \in e} x_u^i$.
	Let $A_v$ denote the event that at most one color in $C_v$ wants one or more nodes in $e$. From Observation~\ref{obs:cv}, $A_v$ is equivalent to the event that  $v$ wants at most 1 color in $C_v$. Thus, the probability of $A_v$ is the probability that at most one of the events $\{X_v^i \colon i \in C_v \}$ happens. Observation~\ref{obs:xuc} gives $\sum_{i = 1}^k \prob[X_v^i] = \sum_{i = 1}^k x_v^i = 1$, allowing us to repurpose a supporting lemma of~\citet{angel2016clustering} on graph \maxecc{} to see that $\prob[A_v] \geq 2/e$ (Lemma~\ref{lem:angel} in Appendix~\ref{app:maxecc}).

	%
	Because color thresholds $\{\alpha_i\}$ are drawn independently for each color, and because color sets $\{C_v \colon v \in e\}$ are disjoint from each other and from $c$, the events $\{A_v, X_v^c \colon v \in e\}$ are mutually independent.
	Thus, using Observation~\ref{obs:edgewants} gives
	\begin{align*}
		\prob\left[ \left(\bigcap_{v \in e} A_v \right) \cap Z_e\right ] & = \prob[Z_e]\cdot \prod_{v \in e} \prob[A_v] \geq z_e \cdot \left(\frac2e \right)^r.
	\end{align*}
	If the joint event $J = (\cap_{v \in e} A_v) \cap Z_e$ holds, this means every node in $e$ wants color $c$, and at most $r$ distinct other colors (one for each node in $v \in e$ since there is one set $C_v$ for each $v \in e$) want one or more nodes in $e$. Conditioned on $J$, $e$ is satisfied if the color $c$ has a higher priority (determined by $\pi$) than the other $r$ colors, which happens with probability $1/(r+1)$. Thus,
	\begin{align*}
		{	\prob[T_e] \geq \prob\left[T_e \mid J\right] \prob\left[ J \right] \geq \frac{z_e}{r+1}\left(\frac2e\right)^r.}
	\end{align*}
\end{proof}

\textbf{Comparison with prior graph \maxecc{} algorithms.}
Algorithm~\ref{alg:hyper_maxecc} achieves a ${4}/({3e^2}) \approx 0.18$ approximation factor for graph \maxecc{} ($r = 2$). This improves upon the first ever approximation factor of $1/e^2 \approx 0.135$ for graph \maxecc{}~\cite{angel2016clustering}, and comes with a significantly simplified proof. There are two interrelated factors driving this simplified and improved guarantee.
The first is our use of Observation~\ref{obs:cv}.
\citet{angel2016clustering} bound the probability of satisfying an edge $(u,v) \in E_c$ using a delicate argument about certain dependent events we can denote by $B_v$ (the event that $v$ wants at most 1 color from $C = [k]\setminus \{c\}$) and $B_u$ (defined analogously for $u$).
The algorithm of \citet{angel2016clustering} requires proving that $\prob[B_{u} \cap B_v] \geq \prob[B_{u}] \prob[B_v]$, even though $B_u$ and $B_v$ are dependent. The proof of this result (Proposition 1 in~\citet{angel2016clustering}) is interesting but also lengthy, and does not work for $r > 2$. Subsequent approximation algorithms for graph \maxecc{} in turn rely on complicated generalizations of this proposition~\cite{ageev2015improved,ageev20200,alhamdan2019approximability}. In contrast, Theorem~\ref{thm:hypermaxecc} deals with mutually independent events $\{A_v \colon v \in e\}$. Observation~\ref{obs:cv} provides the key insight as to why it suffices to consider these events when bounding probabilities, leading to a far simpler analysis that extends easily to hypergraphs.

Our second key factor is the use of a global color ordering $\pi$. \citet{angel2016clustering} and other results for graph \maxecc{} apply a two-stage approach where Stage 1 identifies which nodes ``want'' which colors, and Stage 2 assigns nodes to colors \emph{independently} for each node. To illustrate the difference, consider the probability of satisfying an edge $(u,v) \in E_c$ if we condition on $u$ and $v$ both wanting $c$ and each wanting at most one other color. The algorithm of~\citet{angel2016clustering} has a $1/4$ chance of satisfying the edge (each node gets color $c$ with probability $1/2$), whereas Algorithm~\ref{alg:hyper_maxecc} has a $1/3$ chance (the probability that $c$ is given higher global priority than the other two colors). This is precisely why Algorithm~\ref{alg:hyper_maxecc}'s approximation guarantee is a factor $4/3$ larger than the guarantee of~\citet{angel2016clustering}.

\begin{algorithm}[t]
	\caption{$0.38$-approximation alg. for graph \maxecc{}}
	\label{alg:graph_max_ecc}
	\begin{algorithmic}
		\STATE Obtain optimal variables $\{z_e; x_v^c\}$ for the LP relaxation
		\STATE $\pi \leftarrow \text{uniform random ordering of colors } [k]$
		\STATE For $c \in [k]$, $\alpha_c \leftarrow $ uniform random threshold in $[0,1]$
		\FOR{$v \in V$}
		\STATE $S_v \leftarrow \left\{c \in [k] \mid x_v^c \geq 2/3 \right\}$; $W_v \leftarrow [k] \setminus S_v$
		\STATE $W'_v = \left\{ c \in W_v \mid \alpha_c < x_v^c \right\}$
		\IF{$|W'_v| > 0$}
		\STATE $\lambda(v) \leftarrow \argmax_{i \in W'_v} \pi(i)$
		\ELSIF{$\exists c \text{ s.t. } S_v = \left\{ c \right\}$}
		\STATE $\lambda(v) \leftarrow c$
		\ELSE
		\STATE $\lambda(v) \leftarrow \text{ arbitrary color}$
		\ENDIF
		\ENDFOR
	\end{algorithmic}
\end{algorithm}

\subsection{Graph MaxECC Algorithm}
Although Algorithm~\ref{alg:hyper_maxecc} does not improve on the $0.3622$-approximation of~\citet{ageev20200} for graph \maxecc{}, we can incorporate its distinguishing features (the color ordering $\pi$ and Observation~\ref{obs:cv}) into a refined algorithm with a $154/405 \approx 0.38$ approximation factor.

Our refined algorithm (Algorithm \ref{alg:graph_max_ecc}) for graphs solves the LP relaxation, generates color thresholds $\{\alpha_i \colon i \in [k]\}$, and generates a color ordering $\pi$ in the same way as Algorithm~\ref{alg:hyper_maxecc}. It differs in that it partitions colors for each $v$ into colors that are \emph{strong} or \emph{weak} for $v$, given respectively by the sets
\begin{align*}
	S_v = \{ i \in [k] \colon x_v^i \geq 2/3\} \text{ and } W_v & = [k] - S_v.
\end{align*}
Since $\sum_{i = 1}^k x_v^i = 1$, we know $|S_v| \in \{0,1\}$. We say color $i$ is \textit{strong for $v$} if $S_v = \{i\}$, otherwise $i$ is \textit{weak for $v$}. Note that a color being \emph{strong} or \emph{weak} for $v$ is based on a fixed and non-random LP variable. This is separate from the notion of a color \emph{wanting} $v$, which is a random event.

Algorithm~\ref{alg:graph_max_ecc} first checks if $v$ wants any {weak} colors. If so, it assigns $v$ the weak color of highest priority (using $\pi$) that $v$ wants. If $v$ wants no weak colors but has a strong color, then $v$ is assigned the strong color. Prioritizing weak colors in this way appears counterintuitive, since if $v$ has a strong color $c$ it suggests that $v$ should get color $c$. However, note that $x_v^c \geq 2/3$ still implies $v$ will get color $c$ with high probability, since a large value for $x_v^c$ makes it less likely $v$ will want any weak colors. Meanwhile, prioritizing weak colors enables us to lower bound the probability that an edge $e = (u,v)$ is satisfied even if $\ell(e)$ is weak for $u$ or $v$.

In order to improve on the extensively studied problem of \maxecc{} in graphs, our analysis is much more involved than the proof of Theorem~\ref{thm:hypermaxecc} and requires proving several detailed technical lemmas that may be of interest in their own right. In order to present the lemmas used, we use the following definition.

\begin{definition}
	\label{def:prob}
	Let $m$ be a nonnegative integer and $t$ be an integer. Given a vector $\vx \in \mathbb{R}^m$, if $m \geq t > 0$, we define the function $P(\vx, t)$ as follows:
	\[
		P(\vx, t) = \sum_{I \in {[m] \choose t}} \left( \prod_{i \in I} x_i \prod_{j\in [m] \setminus I} (1 - x_j) \right),
	\]
	and for other choices of $t$ and $m$ we define
	\begin{equation*}
		P(\vx, t) = \begin{cases}
			0                         & \text{ if $m < t$ or $t < 0$} \\
			1                         & \text{ if $0 = m = t$}        \\
			\prod_{i = 1}^m (1 - x_i) & \text{ if $0 = t < m$.}
		\end{cases}
	\end{equation*}
\end{definition}
To provide intuition, $P(\vx,t)$ is defined to encode the probability that $t$ events from a set of $m$ independent events happen. In more detail, consider a set of $m$ mutually independent events $\mathcal{X} = \{X_1, X_2, \cdots, X_m\}$, where the probability that the $i$th event happens is $x_i = \prob[X_i]$, the $i$th entry of $\vx$. The function $P$ in Definition~\ref{def:prob} is the probability that exactly $t$ of the events in $\mathcal{X}$ happen. In our proofs for \maxecc{}, we will apply this with $\mathcal{X}$ representing a subset of colors in the ECC instance, while the vector $\vx$ encodes LP variables $\{x_u^i\}$ for some node $u$ and that set of colors (which by Observation~\ref{obs:xuc} are probabilities for wanting those colors). Lemmas~\ref{lem:bounding} and \ref{lem:bounding_constraints} will aid in bounding the probability that a node is assigned a certain color, conditioned on how many other colors it wants.

The technical lemmas used in the proof of Theorem~\ref{thm:graphecc} are given below. The proofs for these theorems can be found in Appendix~\ref{app:maxecc}.
\begin{restatable}{lemma}{lemdependentynx}
	\label{lem:dependent_y_n_x}
	Consider an edge $e = (u,v)$ of color $c$ and let $c \in W_u$. When running Algorithm~\ref{alg:graph_max_ecc}, let $X_u^c$ be the event that $u$ wants $c$, $Y_u^c$ be the event that $u$ is assigned $c$ by the algorithm, and let $N_v$ be the event that $v$ wants no colors in $[k]\setminus\{c\}$.
	Then the following inequality holds:
	\[
		\prob\left[Y_u^c \mid N_v \cap X_u^c\right] \geq \prob\left[Y_u^c \mid X_u^c\right].
	\]
\end{restatable}

\begin{restatable}{lemma}{lemsumtoprod}
	\label{lem:sum_to_prod}
	Let $\beta \in [0,1]$ be a constant and $\mathbf{x} \in \mathbb{R}^m_{\geq 0}$ be a length $m$ nonnegative vector satisfying $\sum_{t=1}^m x_t \leq \beta$, then we have $\prod_{t=1}^m (1 - x_t) \geq 1 - \beta$.
\end{restatable}

\begin{restatable}{lemma}{lembounding}
	\label{lem:bounding}
	Let $m \geq 2$ be an integer and
	$a_0, a_1, \hdots, a_m$ be values such that for every $t \in [0, m-2]$
	\begin{align*}
		a_{t+1}   & \leq a_t            \\
		2 a_{t+1} & \leq a_t + a_{t+2}.
	\end{align*}
	Let $\mathcal{D} = \{ \vx \in [0,2/3]^m \colon \sum_{i = 1}^m x_i \leq 1 \}$ be the domain of a function
	$f : \mathcal{D} \rightarrow \mathbb{R}$ defined as
	\[
		f(\vx) = \sum_{t=0}^m a_t P(\vx, t).
	\]
	Then $f$ is minimized (over domain $\mathcal{D}$) by any vector $\vx^*$ with one entry set to $2/3$, one entry set to $1/3$, and every other entry set to $0$. Furthermore, we have
	\begin{equation}
		\label{eq:minval}
		f(\vx^*) = \frac{2}{9}(a_0 + a_2) + \frac{5}{9}a_1.
	\end{equation}
\end{restatable}

\begin{restatable}{lemma}{lemboundingconstraints}
	\label{lem:bounding_constraints}
	The inequalities $a_t \geq a_{t+1}$ and $a_t + a_{t+2} \geq 2a_{t+1}$ hold for sequence $a_t = \frac{1}{1+g+t}$ where $g \geq 0$ is an arbitrary fixed integer. These inequalities also hold for the sequence
	\begin{align}
		\label{eq:complicatedat}
		a_t & = \frac{2}{9}\left( \frac{1}{t+1} + \frac{1}{t+3} \right) + \frac{5}{9} \left( \frac{1}{t+2} \right).
	\end{align}
\end{restatable}

Using the above technical lemmas, we prove the following result for an arbitrary edge $e$, which by Observation~\ref{obs:prob} proves our approximation guarantee.
\begin{theorem}
	\label{thm:graphecc}
	For every $e \in E$, $\prob[\text{$e$ is satisfied}] \geq p z_e$ where $p = {154}/{405} > 0.3802$ when running Algorithm~\ref{alg:graph_max_ecc}.
\end{theorem}
\vspace{-10pt}
\begin{proof}

	Fix $e = (u,v)$, let $c = \ell(e)$, and set $C = [k] \backslash \{c\}$.
	Define $C_v = \{i \in C \colon x_v^c \geq x_u^c\}$ and $C_u = C \setminus C_v$. As before, $T_e$ is the event that $e$ is satisfied and $X_v^i$ is the event that $v$ wants color $i$. Let $Y_v^i$ be the event that $v$ is \emph{assigned} color $i$ by Algorithm~\ref{alg:graph_max_ecc}, and $N_v$ be the event that $v$ wants \emph{no} colors in $C$. If $c$ is strong for $v$, event $N_v$ is equivalent to event $Y_v^c$. Let $W_v'$ denote colors in $W_v$ that {want} $v$. Define $X_u^i$, $Y_u^i$, $W_u'$, and $N_u$ analogously for $u$. The proof is separated into (increasingly difficult) cases, based on how many of $\{u,v\}$ have $c$ as a strong color.

	\bigbreak

	\noindent
	\textbf{Case 1:} $c$ is strong for both $u$ and $v$, i.e., $S_u = S_v = \{c\}$.\\
	Let $W_v'$ denote colors in $W_v$ that want $v$, and define $W_u'$ analogously for $u$.
	Since $c$ is the strong color for both $u$ and $v$, $e$ is satisfied if and only if $W'_u = W'_v = \emptyset$. Additionally, because nodes can only have a single strong color and we know $c$ is strong for both $u$ and $v$ we know that all other colors must be weak for both $u$ and $v$. Using the property that probabilities regarding separate colors are independent we have
	\begin{align*}
		\prob\left[e \text{ is satisfied}\right] & = \prob\left[W'_u = W'_v = \emptyset\right]  = \prod_{i \in C}\prob\left[ i \notin W'_u \cup W'_v \right].
	\end{align*}
	Consider a color $i \in C_v$, which by definition means $x_v^i \geq x_u^i$. There are three options for the random threshold $\alpha_i$:
	\begin{itemize}
		\item $\alpha_i < x_u^i \leq x_v^i$, which happens with probability $x_u^i$ and implies that $i \in W'_u \cap W'_v$,
		\item $x_u^i < \alpha_i \leq x_v^i$, which happens with probability $x_v^i - x_u^i$ and implies $i \in W'_v$, $i \notin W'_u$, and
		\item $x_u^i \leq x_v^i < \alpha_i$, which happens with probability $1 - x_v^i$ and implies $i \notin W_u' \cup W_v'$.
	\end{itemize}
	Thus, for $i \in C_v$ we have $\prob\left[i \notin W'_u \cup W'_v\right] = (1 - x_v^i)$ and similarly for $i \in C_u$ we have $\prob\left[i \notin W'_u \cup W'_v\right] = (1 - x_u^i)$. This gives
	\begin{align*}
		\prod_{i \in C}\prob\left[ i \notin W'_u \cup W'_v \right] & = \prod_{i \in C_v}(1 - x_v^i)\prod_{i \in C_u}(1 - x_u^i).
	\end{align*}
	The fact that $c$ is strong for both $u$ and $v$ means $x_v^i \geq 2/3$ and $x_u^i \geq 2/3$. From the equality constraint in the LP we see that for $w\in\{u,v\}$ we have $\sum_{i \in C}x_w^i \leq 1 - 2/3 = 1/3$. Applying Lemma \ref{lem:sum_to_prod} gives
	\begin{align*}
		\prod_{i \in C_v}(1 - x_v^i)\prod_{i \in C_u}(1 - x_u^i) & \geq \frac{2}{3}\frac{2}{3} = \frac{4}{9} \geq \frac{4}{9}z_e.
	\end{align*}


	\bigbreak

	\noindent
	\textbf{Case 2:} $c$ is strong for one of $u$ or $v$. \\
	Without loss of generality we say $S_v = \{c\}$ and $c \in W_u$.
	Edge $e$ is satisfied if and only if $Y_u^c \cap Y_v^c$ holds.
	Because $c$ is strong for $v$, event $Y_v^c$ holds if and only if $N_v$ holds. Using the fact that $Y_u^c = Y_u^c \cap X_u^c$, we can write
	\[
		\prob\left[e \text{ is satisfied}\right] = \prob\left[Y_u^c \cap X_u^c \cap N_v\right].
	\]
	Using Bayes' Theorem, the fact that $X_u^c$ and $N_v$ are independent\footnote{Independence follows from the fact that $X_u^c$ is concerned with color $c$, while $N_v$ is concerned with a disjoint color set $C = [k]\backslash \{c\}$.}, Observation~\ref{obs:xuc} ($\prob[X_u^c] = x_u^c$), and Lemma~\ref{lem:dependent_y_n_x}, we see that
	\begin{align*}
		\prob\left[Y_u^c \cap X_u^c \cap N_v\right] & = \prob\left[ Y_u^c \mid N_v \cap X_u^c \right] \prob\left[ N_v \cap X_u^c \right]                & \text{(Bayes' Theorem)}                     \\
		                                            & = \prob\left[ Y_u^c \mid N_v \cap X_u^c \right] \prob\left[ N_v \right] \prob\left[ X_u^c \right] & \text{($X_u^c$ and $N_v$ are indepdendent)} \\
		                                            & = \prob\left[ Y_u^c \mid N_v \cap X_u^c \right] \prob\left[ N_v \right] x_u^c                     & \text{(Observation~\ref{obs:prob})}         \\
		                                            & \geq \prob\left[ Y_u^c \mid N_v \cap X_u^c \right] \prob\left[ N_v \right] z_e                    & \text{(LP constraint $x_u^c \geq z_e$)}     \\
		                                            & \geq \prob\left[ Y_u^c \mid X_u^c \right] \prob\left[ N_v \right] z_e                             & \text{(Lemma~\ref{lem:dependent_y_n_x}).}
	\end{align*}
	Using $\overline{X}_v^i$ to indicate that $v$ \emph{does not} want $i$, we get
	\[
		\prob\left[ N_v \right] = \prob\left[ \bigcap_{i \in C} \overline{X}_v^i \right] = \prod_{i \in C}(1 - x_v^i).
	\]
	As in Case 1, because $c$ is strong for $v$, we know that $\sum_{i \in C} x_v^i \leq 1/3$ and applying Lemma \ref{lem:sum_to_prod} gives us that $\prob\left[ N_v \right] \geq 2/3$.

	Finally we must bound $\prob\left[ Y_u^c \mid X_u^c \right]$. Since
	$c$ is weak for $u$ (i.e., $c \in W_u$), it will be convenient to consider all weak colors for $u$ other than $c$, which we will denote by $\hat{W}_u = W_u \backslash \{c\}$. We will then use $\hat{W}_u' = \{ i \in \hat{W}_u \colon \alpha_i < x_u^i\}$ to denote the set of colors in $\hat{W}_u$ that $u$ wants. Note that $X_u^c$ is independent of the number of colors in $\hat{W}_u$ that $u$ wants.

	Because of the global ordering of colors, conditioned on $u$ wanting $c$ and wanting exactly $t$ colors in $\hat{W}_u$, there is a $1/(t+1)$ chance that $c$ is chosen first by the permutation $\pi$, and is therefore assigned to $u$. Formally:
	\begin{equation*}
		\prob\left[ Y_u^c \mid X_u^c  \cap (|\hat{W}'_u| = t)\right] = \frac{1}{t+1}.
	\end{equation*}
	Define $p_t = \prob[ |\hat{W}'_u| = t]$ to be the probability that exactly $t$ colors from $\hat{W}_u$ are wanted by $u$. Using the vector $\vx_u[\hat{W}_u]$ to encode LP variables $\{x_u^i \colon i \in \hat{W}_u\}$ for node $u$ and the colors in $\hat{W}_u$, we see that $p_t = P(\vx_v[\hat{W}_u], t)$. Therefore, the law of total probability (combined with the fact that $X_u^c$ is independent from $|\hat{W}_u'|$) gives
	\begin{equation}
		\label{eq:mainprob}
		\prob\left[ Y_u^c \mid X_u^c \right]
		= \sum_{t = 0}^{|\hat{W}_u|} \prob\left[ Y_u^c \mid X_u^c  \cap (|\hat{W}'_u| = t)\right] \prob[|\hat{W}'_u| = t] = \sum_{t = 0}^{|\hat{W}_u|} \frac{1}{1+t} p_t.
	\end{equation}
	Eq.~\eqref{eq:mainprob} is exactly of the form in Lemma \ref{lem:bounding} with $a_t = 1/(t+1)$. From Lemma \ref{lem:bounding_constraints}, we know this sequence $\{a_t\}$ satisfies the constraints needed by Lemma \ref{lem:bounding}, and applying Lemma~\ref{lem:bounding} shows $\prob\left[ Y_u^c \mid X_u^c \right] \geq 31/54$. Combining this with the previously established bound $\prob\left[ N_v \right] \geq 2/3$ gives
	\begin{align*}
		\prob\left[e \text{ is satisfied}\right] & \geq \prob\left[ Y_u^c \mid X_u^c \right] \prob\left[ N_v \right] z_e  \geq \frac{31}{54} \cdot \frac{2}{3} z_e = \frac{31}{81}z_e > \frac{154}{405} z_e.
	\end{align*}


	\bigbreak

	\noindent
	\textbf{Case 3:} $c$ is weak for both $u$ and $v$, i.e., $c \in W_u \cap W_v$.\\
	First we condition the probability of satisfying $e$ on the event $X_u^c \cap X_v^c$. If one of $X_u^c$ or $X_v^c$ does not happen, then $e$ cannot be satisfied, so we have
	\begin{align*}
		\prob\left[ \text{$e$ is satisfied} \right] & = \prob\left[ \text{$e$ is satisfied} \mid X_u^c \cap X_v^c \right] \prob\left[ X_v^c \cap X_u^c \right] = \prob\left[ \text{$e$ is satisfied} \mid X_u^c \cap X_v^c \right] z_e.
	\end{align*}
	Define $D = (W_u \cup W_v) \backslash \{c\}$ to be the set of weak colors (excluding $c$) that want at least one of $\{u,v\}$. If we condition on $u$ and $v$ both wanting $c$, it is still possible that $e$ will not be satisfied. This will happen if a weak color $i \in D$ wants either $u$ or $v$, and the permutation $\pi$ prioritizes color $i$ over $c$. We bound the probability of satisfying $e$ by conditioning on the number of colors in $D$ that want one or both of $\{u,v\}$.

	Using the same logic as in Observation~\ref{obs:cv}, we can present a simpler way to characterize whether a color $i \in D$ wants at least one of $\{u,v\}$. Formally, we partition $D$ into the sets:
	\begin{align*}
		D_u & = \{i \in D \colon x_u^i \geq x_v^i \} \\
		D_v & = D - D_v.
	\end{align*}
	Note that a color $i \in D_u$ wants one or both nodes in $e$ if and only if $i$ wants $u$, and color $i \in D_v$ wants one or both nodes in $e$ if and only if $i$ wants $v$.
	Mirroring our previous notation, let $D_u'$ denote the set of colors in $D_u$ that want node $u$ (based on the random color thresholds) and define $D_v'$ analogously. The set $D_u' \cup D_v'$ is then the set of weak colors that want one or more node in $e$. Thus, conditioned on $|D_u' \cup D_v'| = t$ and conditioned on $u$ and $v$ wanting $c$, in order for $e$ to be satisfied $c$ must come before all $t$ colors in $|D_u' \cup D_v'|$ in the random permutation $\pi$. Formally, this means
	\[
		\prob\left[ \text{$e$ is satisfied} \mid X_u^c \cap X_v^c \cap (|D'_u \cup D'_v| = t) \right] = \frac{1}{1 + t}.
	\]

	Again using the law of total probability and the fact that $X_u^c \cap X_v^c$ is independent from $|D_u' \cup D_v'|$, we see
	\begin{align}
		\label{eq:dupdvp}
		\prob[ e & \text{ is satisfied} \mid X_u^c \cap X_v^c ]
		= \sum_{t = 0}^{|D|} \frac{1}{1+t} \prob\left[ |D'_u \cup D'_v| = t \right].
	\end{align}
	Observe now that $D_u$ and $D_v$ are disjoint color sets, which implies that $|D_u'|$ and $|D_v'|$ are independent random variables. This allows us to decouple $D_u'$ and $D_v'$ in the above expression since $\prob\left[ |D'_u \cup D'_v| = t \right] = \prob\left[ |D'_u| + |D'_v| = t \right]$. We define
	\[
		p_i = \prob\left[ |D'_u| = i \right] \text{ and } q_i = \prob\left[ |D'_v| = i \right],
	\]
	so that we can write
	\[
		\prob\left[ |D'_u| + |D'_v| = t \right] = \sum_{\substack{i, j \in [0, t] \\ i + j = t}} p_i q_j = \sum_{i = 0}^{|D_u|} p_i q_{t-i} .
	\]
	This allows us to re-write Eq.~\eqref{eq:dupdvp} as
	\begin{align}
		\label{eq:decoupled}
		\prob[ e \text{ is satisfied} \mid X_u^c \cap X_v^c ] & = \sum_{i = 0}^{|D_u|} \sum_{j = 0}^{|D_v|} \frac{1}{1+i+j} p_i q_j.
	\end{align}

	Without loss of generality we now assume $|D_u| \leq |D_v|$, and proceed case by case depending on the sizes of $D_u$ and $D_v$.
	Our goal is to show that for all cases $\prob[ e \text{ is satisfied} \mid X_u^c \cap X_v^c ] \geq 154/405$.
	In the cases below, we use the vector $\vx_u[D_u]$ to encode LP variables for node $u$ and colors in $D_u$ so that $p_i = P(\vx_u[D_u], i)$, and likewise $q_i = P(\vx_v[D_v], i)$, to match with notation in Definition~\ref{def:prob} and Lemma~\ref{lem:bounding}.

	\bigbreak

	\noindent
	\textbf{Case 3a.} $|D_u| = |D_v| = 0$. \\
	No weak colors other than $c$ want $u$ or $v$, so $\prob[ e \text{ is satisfied} \mid X_u^c \cap X_v^c ] = 1$.

	\bigbreak

	\noindent
	\textbf{Case 3b.} $|D_u| = 0$ and $|D_v| = 1$. \\
	Letting $a$ be the single color in $D_v$,
	\begin{align*}
		\prob[ e \text{ is satisfied} \mid X_u^c \cap X_v^c ] & = p_0 q_0 + \frac{1}{2} p_0 q_1 = q_0 + \frac{1}{2} q_1 = (1 - x_v^a) + \frac{1}{2} x_v^a.
	\end{align*}
	Because $a$ is a weak color for $v$, we know $x_v^a \in [0, \frac{2}{3}]$, and the minimum value the probability can obtain is $\frac{2}{3} > \frac{154}{405}$.

	\bigbreak

	\noindent
	\textbf{Case 3c.} $|D_u| = |D_v| = 1$. \\
	Letting $a$ be the color in $D_v$ and $b$ be the color in $D_u$,
	\begin{align*}
		\prob[ e \text{ is satisfied} \mid X_u^c \cap X_v^c ] & = p_0 q_0 + \frac{1}{2}(p_0 q_1 + p_1 q_0) + \frac{1}{3} p_1 q_1 \\ &= (1 - x_u^b)(1 - x_v^a) + \frac{1}{2}((1 - x_u^b)x_v^a + x_u^b(1 - x_v^a)) + \frac{1}{3} x_u^b x_v^a.
	\end{align*}
	Because $a$ is weak for $v$ and $b$ is weak for $u$, we know $x_v^a, x_u^b \in [0, \frac{2}{3}]$. This gives a minimum value for the probability of $\frac{13}{27} > \frac{154}{405}$.

	\bigbreak

	\noindent
	\textbf{Case 3d.} $|D_u| = 0$ and $|D_v| > 1$. \\
	In this case $p_0 = 1$, and Eq.~\eqref{eq:decoupled} simplifies to
	\begin{equation*}
		\prob[ e \text{ is satisfied} \mid X_u^c \cap X_v^c ] = \sum_{j = 0}^{|D_v|} \frac{1}{j+1} q_j \geq \frac{31}{54},
	\end{equation*}
	where we have applied Lemma \ref{lem:bounding} with $a_j = \frac{1}{j+1}$.

	\bigbreak

	\noindent
	\textbf{Case 3e.} $|D_u| = 1$ and $|D_v| > 1$. \\
	Eq.~\eqref{eq:decoupled} simplifies to
	\begin{align*}
		\prob[ e \text{ is satisfied} \mid X_u^c \cap X_v^c ] & = p_0 \sum_{j = 0}^{|D_v|} \frac{q_j}{j+1}  + p_1 \sum_{j = 0}^{|D_v|} \frac{q_j}{j+2} \geq p_0 \frac{31}{54} + p_1 \frac{19}{54},
	\end{align*}
	where we applied Lemma \ref{lem:bounding} twice, once with $a_j = \frac{1}{j+1}$ and once with $a_j = \frac{1}{j+2}$. The right hand side of the above bound has a minium value of $\frac{23}{54}$.

	\bigbreak

	\noindent
	\textbf{Case 3f.} $|D_u| > 1$ and $|D_v| > 1$.
	From Eq.~\eqref{eq:decoupled} we know
	\begin{align*}
		\prob[ e \text{ is satisfied} \mid X_u^c \cap X_v^c ] & = \sum_{i = 0}^{|D_u|}  p_i \left(\sum_{j = 0}^{|D_v|} \frac{1}{1+i+j}  q_j \right).
	\end{align*}
	We then apply Lemma~\ref{lem:bounding} $|D_u|+1$ times, once for each choice of $i \in \{0,1, \hdots, |D_u|\}$. The $i$th time we apply it, we use the sequence $a_j = \frac{1}{1+i+j}$ to get
	\begin{align}
		\label{eq:lemdvtimes}
		\sum_{j = 0}^{|D_v|} \frac{1}{1+i+j}  q_j \geq \frac{2}{9}\left( \frac{1}{i+1} + \frac{1}{i+3} \right) + \frac{5}{9} \left( \frac{1}{i+2} \right).
	\end{align}
	The right hand of~\eqref{eq:lemdvtimes} defines a new sequence
	\begin{equation}
		\label{eq:uglyai}
		a_i = \left( \frac{1}{i+1} + \frac{1}{i+3} \right) + \frac{5}{9} \left( \frac{1}{i+2} \right).
	\end{equation}
	We know from Lemma~\ref{lem:bounding_constraints} that this sequence $\{a_i\}$ in Eq.~\eqref{eq:uglyai} satisfies the conditions from Lemma~\ref{lem:bounding}, so applying Lemma~\ref{lem:bounding} one more time and combining all steps gives
	\begin{align*}
		\prob[ e \text{ is satisfied} \mid X_u^c \cap X_v^c ] & = \sum_{i = 0}^{|D_u|}  p_i \left(\sum_{j = 0}^{|D_v|} \frac{1}{1+i+j}  q_j \right) \\
		&\geq \sum_{i = 0}^{|D_u|}  p_i \left( \frac{2}{9}\left( \frac{1}{i+1} + \frac{1}{i+3} \right) + \frac{5}{9} \left( \frac{1}{i+2} \right) \right) \geq \frac{154}{405}.
	\end{align*}
	Therefore, we have shown in all possible subcases of Case 3 that $\prob\left[ \text{$e$ is satisfied} \mid X_u^c \cap X_v^c \right] \geq \frac{154}{405}$.

	Cases 1,2, and 3 all together show that for every $e \in E$ we have
	\[
		\prob\left[ \text{$e$ is satisfied} \right] \geq \frac{154}{405} z_e.
	\]
\end{proof}

\section{Alternative ECC Objectives}
We now turn our attention to more general~\ECC{} objectives, as well as several which incorporate notions of balance or fairness with respect to edge colors.
\subsection{Generalized Norm Objectives}
\label{sec:generalized-mean}
We begin by introducing~\pmeanECC{}, which we define via the following binary program:
\begin{align}
	\label{eq:p-mean-ecc}
\begin{aligned}
	\text{min} \quad& \textstyle (\sum_{c = 1}^k (m_c)^p)^{\frac{1}{p}} \\
	\text{s.t.} \quad& \forall v \in V:\\
	&\forall c \in [k], e \in E_c:\\
    &\forall c \in [k]: \\
	&d_v^c, \gamma_e \in \{0, 1\}
\end{aligned}
\begin{aligned}
	&\\
	&\textstyle \sum_{c=1}^k d_v^c \geq k - 1 \\
	\quad &d_v^c \leq \gamma_e \quad \forall v \in e \\
    & m_c = \textstyle\sum_{e\in E_c}\gamma_e\\
	&\forall c \in [k], v \in V, e \in E.
\end{aligned}
\end{align}
Recalling the notation of BLP~\eqref{eq:maxecc}, we think of $d_v^c = 1 - x_v^c$ as encoding the ``distance'' of node $v$ to color $c$, and $\gamma_e = 1 - z_e$ as encoding whether hyperedge $e$ is \emph{unsatisfied} ($\gamma_e = 1$) or not ($\gamma_e = 0$).
The variables $m_c$ can be thought of as the entries in a per-color error vector $\textbf{m} \in \mathbb{Z}^k$ which tracks the number of unsatisfied edges of each color.
Otherwise, the variables and constraints are similar to BLP~\eqref{eq:maxecc}, adapted for the minimization objective. 
When $p = 1$, we recover the standard \minecc{} problem, i.e., the $\ell_1$-norm on $\textbf{m}$.
We now show how to approximate the generalized \pmeanECC{} problem.
For $p \geq 1$, if we relax the binary constraints to linear constraints $0 \leq d_v^c \leq 1$ and $0 \leq \gamma_e \leq 1$, we obtain a convex programming relaxation. We can then apply convex optimization techniques to find the optimal solution to this relaxation and round the resulting fractional coloring.
The rounding scheme is similar to that used for \minecc{}~\cite{amburg2020clustering}.
In particular, given a fractional coloring consisting of optimal variables $\{d_v^c\}$, we assign color $c$ to vertex $v$ if and only if $d_v^c < \frac{1}{2}$.
We show that the resulting coloring is well-defined and a $2$-approximation to the generalized objective.
When $0 < p < 1$ we lose convexity, but we can still recover approximations; we give a classic analysis based on the Lov{\'a}sz extension~\cite{lovasz1983submodular}.
Some readers may find it helpful to observe that our analysis is conceptually equivalent to the textbook $2$-approximation for \textsc{Submodular Vertex Cover}\footnote{Given a graph $G = (V, E)$ and a submodular function $f\colon 2^V \rightarrow \mathbb{R}_{\geq 0}$, find a vertex cover $C$ of $G$ which minimizes $f(C)$.}, where the graph in question is the \emph{conflict graph} $G$ of our edge-colored hypergraph $H$: for each hyperedge $e$ create a vertex $v_e$, and add the edge $(v_e,v_f)$ for each pair $e, f$ of hyperedges with $\ell(e) \neq \ell(f)$ and $e \cap f \neq \emptyset$.
The connection between vertex colorings of edge-colored hypergraphs and vertex covers of their associated conflict graphs has been observed several times in the literature~\cite{angel2016clustering,cai2018alternating,kellerhals2023parameterized,veldt2023optimal}, and also appears elsewhere in the present work, e.g., in the proofs of~\Cref{thm:cfminecc-reduce-sparse-vc,thm-combinatorial-k-approx,thm:pc-multiple-classes-NP-hard}.
\begin{restatable}{theorem}{pmeanapproximation}\label{thm:pmean-approximation}
    \pmeanECC{} is approximable within factor $2$ when $p \geq 1$, and factor $2^{1/p}$ when $0 < p < 1$.
\end{restatable}
\begin{proof}

    We begin with the case where $p \geq 1$.
    It is well known that the $\ell_p$-norm is convex, and we observe that all of the constraints
    in Program~(\ref{eq:p-mean-ecc}) are linear. Thus, we may in polynomial time obtain optimal fractional variables $\{d_v^c\}, \{\gamma_e\}, \{m_c\}$
    for the relaxation of Program~(\ref{eq:p-mean-ecc}).

    We observe that for each vertex $v$, there is at most one color $c$ with the property that $d_v^c < \frac{1}{2}$.
    Otherwise, there are two colors $c_1, c_2$ with $d_v^{c_1}, d_v^{c_2} < \frac{1}{2}$, but then
    \[
        \sum_{c=1}^k d_v^c < 2\cdot\frac{1}{2} + (k - 2) = k - 1,
    \]
    contradicting the first constraint of Program~(\ref{eq:p-mean-ecc}).

    We now propose a coloring $\lambda$ as follows. For each color $c$ and vertex $v$, we assign $\lambda(v) = c$ if $d_v^c < \frac{1}{2}$.
    By the preceding observation, this procedure assigns at most one color to each vertex.
    If for any $v$ we have that $d_v^c \geq \frac{1}{2}$ for all colors $c$, then we set $\lambda(v)$ arbitrarily.

    Now, for each hyperedge $e$ let $\hat{\gamma}_e$ be equal to $0$ if $\lambda(v) = \ell(e)$ for all $v \in e$, or $1$ otherwise.
    We claim that if $\hat{\gamma}_e = 1$, then $\gamma_e \geq \frac{1}{2}$.
    Otherwise, letting $c = \ell(e)$, the second constraint of Program~(\ref{eq:p-mean-ecc}) implies that $d_v^c < \frac{1}{2}$ for all $v \in e$.
    Then $\lambda(v) = c$ for all $v \in e$, implying that $\hat{\gamma}_e = 0$, a contradiction.
    We conclude that for every hyperedge $e$, $\hat{\gamma}_e \leq 2\gamma_e$.

    For each color $c$, we write $\hat{m}_c$ for the number of hyperedges of color $c$ which are unsatisfied by $\lambda$.
    Equivalently,
    \[
        \hat{m}_c = \sum_{e \in E_c} \hat{\gamma}_e \leq 2\sum_{e \in E_c} \gamma_e = 2m_c,
    \]
    where the last equality comes from the third constraint of Program~\ref{eq:p-mean-ecc}.

    It follows that for every $c$, $\hat{m}_c^p \leq 2^pm_c^p$, and thus
    \[
        \sum_{c=1}^k \hat{m}_c^p \leq 2^p \cdot \sum_{c=1}^k m_c^p,
    \]

    which in turn implies that

    \[
        \left(\sum_{c=1}^k \hat{m}_c^p\right)^{1/p} \leq 2 \cdot \left(\sum_{c=1}^k m_c^p\right)^{1/p},
    \]

    so $\lambda$ is $2$-approximate.

    We now consider the case where $0 < p < 1$.
    We begin by defining, for each $c \in [k]$, a non-negative, monotone, and submodular function $f_c\colon 2^E \rightarrow \mathbb{Z}$ given by $f_c(S) = |E_c \cap S|$.
    If $\lambda$ is a vertex coloring of $H$ and $S_\lambda \subseteq E$ is the set of edges unsatisfied by $\lambda$, then $f_c(S_\lambda)$ measures the number of edges of color $c$ unsatisfied by $\lambda$.
    Because the composition of a concave function with a submodular function results in a submodular function and $0 < p < 1$, $f_c^p$ is non-negative, monotone, and submodular.
    Then the function $f\colon 2^E \rightarrow \mathbb{R}$ given by the sum
    \[
        f(S) = \sum_{c=1}^{k} f_c^p(S)
    \]
    has these same properties.

    Now, let $\lambda^*$ be an optimal vertex coloring of $H$ and let $S_{\lambda^*}$ be the set of hyperedges unsatisfied by $\lambda^*$.
    Our goal will be to compute a coloring $\lambda$ with $f(S_\lambda) \leq 2f(S_{\lambda^*})$.
    This would complete the proof, since then we have
    \[
        (f(S_\lambda))^{1/p} \leq (2f(S_{\lambda^*}))^{1/p} = 2^{1/p}f^{1/p}(S_{\lambda^*}),
    \]
    and $f^{1/p}(S_{\lambda}), f^{1/p}(S_{\lambda^*})$ are precisely the objective values attained by $\lambda$ and $\lambda^*$, respectively, according to Program~(\ref{eq:p-mean-ecc}).


    Impose an arbitrary order on the edges $E$ of our edge-colored hypergraph $H$, and for each vector $\boldsymbol{\gamma} \in [0, 1]^{|E|}$ map edges to components of $\boldsymbol{\gamma}$ according to this order.
    Henceforth, for a hyperedge $e$ we write $\gamma_e$ for the component of $\boldsymbol{\gamma}$ corresponding to $e$. Moreover, for a
    value $\rho$ selected uniformly at random from $[0, 1]$, we write $T_\rho(\boldsymbol{\gamma}) = \{e \colon \gamma_e \geq \rho\}$.
    Then the Lov{\'a}sz extension $\hat{f}\colon [0, 1]^{|E|} \rightarrow \mathbb{R}$ of $f$ is given by
    \[
        \hat{f}(\boldsymbol{\gamma}) = \mathbb{E}[f(T_\rho(\boldsymbol{\gamma}))].
    \]
    This $\hat{f}$ is convex~\cite{lovasz1983submodular}, so we may efficiently optimize the following program:
    \begin{align}
        \label{eq:lovasz}
        \begin{aligned}
            \text{min} \quad  & \hat{f}(\gamma)                 \\
            \text{s.t.} \quad & \gamma_e + \gamma_f \geq 1\quad \\
                              & \gamma_e \geq 0                 \\
        \end{aligned}
        \begin{aligned}
             &                                                                                                         \\
             & \text{for all } e, f \in E \text{ such that } e \cap f \neq \emptyset \text{ and } \ell(e) \neq \ell(f) \\
             & \text{for all } e \in E.
        \end{aligned}
    \end{align}
    Consider the binary vector given by $\gamma_e = 1$ if $e \in S_{\lambda_*}$ or $0$ otherwise; note that this vector satisfies the constraints of Program~(\ref{eq:lovasz}). Hence, we can see that the minimum value for Program~(\ref{eq:lovasz}) provides a lower bound on $f(S_{\lambda_*})$.

    Let $\boldsymbol{\gamma}$ be a minimizer of Program~(\ref{eq:lovasz}), and let
    \[
        S = \left\{e \in E \colon \gamma_e \geq \frac{1}{2}\right\}.
    \]
    The first constraint of Program~\ref{eq:lovasz} ensures that $S$ contains at least one member of every pair of overlapping and distinctly colored hyperedges.
    Thus, it is trivial to compute a coloring $\lambda$ with $S_\lambda \subseteq S$.
    It follows from monotonicity that $f(S_\lambda) \leq f(S)$, and that for each $\rho \leq \frac{1}{2}$, $f(S) \leq f(T_\rho(\boldsymbol{\gamma}))$.
    Putting it all together, we have

    \[
        \frac{f(S_\lambda)}{2} \leq \frac{f(S)}{2} = \int_{0}^{\frac{1}{2}} f(S) \ \mathrm{d}\rho \leq \int_{0}^{\frac{1}{2}} f(T_\rho(\boldsymbol{\gamma})) \ \mathrm{d}\rho \leq \int_{0}^{1} f(T_\rho(\boldsymbol{\gamma})) \ \mathrm{d}\rho = \hat{f}(\boldsymbol{\gamma}) \leq f(S_{\lambda^*}),
    \]
    which completes the proof.

\end{proof}

As $p$ increases, the $\ell_p$-norm on $\textbf{m}$ converges to the maximum component value in $\textbf{m}$. 
Thus, for $p \rightarrow \infty$ we recover a mini-max variant of \ECC{}, for which the objective is to minimize the maximum number of unsatisfied edges of any color.
We will discuss this variant in more detail in~\Cref{sec:color-fair}, but first we pause to note that this convergence of the $\ell_p$ objective
implies that we cannot hope to obtain a better-than $2$-approximation by rounding fractional solutions to Program~(\ref{eq:p-mean-ecc}).
We need only consider a triangle on three vertices $v_1, v_2, v_3$ with distinct edge colors $c_{12}, c_{23}, c_{13}$. The optimal mini-max objective value is $1$, but
a fractional relaxation of Program~(\ref{eq:p-mean-ecc})
can achieve an objective value of $\frac{1}{2}$ by setting $d_{v_1}^{c_{12}} = d_{v_1}^{c_{13}} = d_{v_2}^{c_{12}} = d_{v_1}^{c_{23}} = d_{v_3}^{c_{13}} = d_{v_3}^{c_{23}} = \frac{1}{2}$.

\begin{restatable}{theorem}{pmeanintegralitygap}
    As $p \rightarrow \infty$, the integrality gap of Program~(\ref{eq:p-mean-ecc}) converges to $2$. 
\end{restatable}

\subsection{Color-fair ECC variants}\label{sec:color-fair}

Now we focus on the mini-max case of~\pmeanECC{} given by $p \rightarrow \infty$, which we refer to as \cfminECC{}.
In the following, $M_c^\lambda$ indicates the number of edges of color $c$ which are unsatisfied by a coloring $\lambda$.

\problembox{\cfminECC{}}{A $k$-edge-colored hypergraph $H = (V, E, \ell)$ and an integer $\tau$.}{Does there exist a vertex coloring $\lambda$ of $H$ with $\max_{c \in [k]} M_c^\lambda \leq \tau$?}

The optimization problem asks for a coloring $\lambda$ which minimizes $\max_{c \in [k]} M_c^\lambda$.
Also of interest is the corresponding maxi-min variant~\cfmaxECC{}, for which the question is whether there exists a coloring which satisfies at least $\tau$ edges of \emph{every} color.
These problems are not equivalent at optimality in general, though they are when restricted to hypergraphs with an equal number of edges of every color.
We begin by establishing hardness for both problems.
%
\begin{restatable}{theorem}{cfmineccNPhard}\label{thm:cfminecc-NPhard}
    \cfminECC{} is \cclass{NP}-hard even when restricted to subcubic trees with cutwidth~$2$, exactly $3$ edges of each color, and $\tau = 2$. \cfmaxECC{} is \cclass{NP}-hard even when restricted to paths with $\tau = 1$.
\end{restatable}
\begin{proof}[Proof sketch.]
    Here, we give a short reduction which proves only the following weaker claim:
    
    \begin{claim}
        \cfminECC{} is \cclass{NP}-hard even when restricted to hypergraphs with exactly three edges of each color, and $\tau = 2$. \cfmaxECC{} is \cclass{NP}-hard on the same class of hypergraphs with $\tau = 1$.
    \end{claim}
    
    The full proof, which shows all of the structural restrictions in the theorem statement, uses essentially the same reduction, but with additional gadgeteering. We refer the interested reader to~\Cref{appendix:color-fair}.

    \bigskip
    We reduce from \threeSAT{}, for which the input is a formula written in conjunctive normal form, consisting of $m$ clauses $C_1, C_2, \ldots, C_m$ over $n$ boolean variables $x_1, x_2, \ldots x_n$.
    For each variable $x_i$, we write $x_i$ and $\neg x_i$ for the corresponding positive and negative literals. We assume that every clause contains three literals corresponding to three distinct variables, and that for each variable $x_i$, both literals of $x_i$ appear in the formula.
    We construct an instance $(H = (V, E), \tau = 2)$ of \cfminECC{} as follows.

    For each clause $C_j$, we create a unique color $c_j$.
    Next, for each variable $x_i$, we create a vertex $v_{j_1, j_2}^i$ for each (ordered) pair of clauses $C_{j_1}, C_{j_2}$ with $C_{j_1}$ containing the positive literal $x_i$ and $C_{j_2}$ containing the negative literal $\neg x_i$.
    We say that $v_{j_1, j_2}^i$ is associated with the variable $x_i$ and the variable-clause pairs $(x_i, C_{j_1})$ and $(x_i, C_{j_2})$. Observe that if variable $x_i$ is contained in clause $C_j$, then $(x_i, C_j)$ has at least one associated vertex.

    Now we create edges. For each clause $C_j$ we create one edge $e_j^i$ of color $c_j$ for each variable $x_i$ contained in $C_j$. This edge contains every vertex associated with the variable-clause pair $(x_i, C_j)$.
    We note that some hyperedges may have size one. To exclude this possibility, we need only introduce an auxiliary vertex for each such hyperedge, which is contained in no other edges.
    We say that $e_j^i$ is \emph{positive} for $x_i$ if the positive literal $x_i$ appears in $C_j$, and \emph{negative} for $x_i$ otherwise.
    We set $\tau = 2$.
    This completes the construction.
    Observe that we have created exactly three hyperedges of each color, and that each hyperedge contains vertices associated with exactly one variable.
    It follows that hyperedges $e$ and $f$ overlap if and only if they are distinctly colored, contain vertices associated with the same variable $x_i$, and (without loss of generality) $e$ is positive for $x_i$ while $f$ is negative for $x_i$.   

    It remains to show that the reduction is correct.
    For the first direction, assume that there exists an assignment $\phi$ of boolean values to the variables $x_1, x_2, \ldots, x_n$ which satisfies every clause. We say that the assignment $\phi$ \emph{agrees} with a variable-clause pair $(x_i, C_j)$ if $C_j$ contains the positive literal $x_i$ and $\phi(x_i) = \textsf{True}$ or if $C_j$ contains the negative literal $\neg x_i$ and $\phi(x_i) = \textsf{False}$.
    We color the vertices of our constructed graph as follows.
    For each vertex $v_{j_1, j_2}^i$, if $\phi$ agrees with $(x_i, C_{j_1})$ we assign color $c_{j_1}$ to $v_{j_1, j_2}^i$, and otherwise we assign color $c_{j_2}$.
    Observe that this coloring satisfies an edge $e_j^i$ if and only if $\phi$ agrees with $(x_i, C_j)$.
    Moreover, because $\phi$ is satisfying,
    every clause $C_j$ contains at least one variable $x_i$ such that $\phi$ agrees with $(x_i, C_j)$.
    Hence, at least one edge of every color is satisfied. Because there are exactly three edges of every color, there are at most two unsatisfied edges of any color.

    For the other direction, assume that we have a coloring which leaves at most two edges of any color unsatisfied. We will create a satisfying assignment $\phi$.
    For each variable $x_i$, we set $\phi(x_i) = \textsf{True}$ if any edge which is positive for $x_i$ is satisfied, and $\phi(x_i) = \textsf{False}$ otherwise.
    We now show that $\phi$ is satisfying. Consider any clause $C_j$. There are exactly three edges with color $c_j$, and at least one of them is satisfied.
    Let $x_i$ be the corresponding variable, so the satisfied edge is $e_j^i$.
    If $C_j$ contains the positive literal $x_i$, then $e_j^i$ is positive for $x_i$ and so $\phi(x_i) = \textsf{True}$. Hence, $C_j$ is satisfied by $\phi$.
    Otherwise $C_j$ contains the negative literal $\neg x_i$, and $e_j^i$ is negative for $x_i$.
    In this case, we observe that every edge which is positive for $x_i$ intersects with $e_j^i$, and none of these edges has color $c_j$ since
    $C_j$ does not contain both literals. Hence, the satisfaction of $e_j^i$ implies that every edge which is positive for $x_i$ is unsatisfied.
    It follows that $\phi(x_i) = \textsf{False}$, meaning $C_j$ is satisfied by $\phi$.

    The result for~\cfmaxECC{} follows from the same construction, setting $\tau = 1$ instead of~$2$. Because there are exactly three edges of every color, the analysis is identical.
\end{proof}
\Cref{thm:cfminecc-NPhard} implies that it is \cclass{NP}-hard to provide any multiplicative approximation for \cfmaxECC{}, since it is \cclass{NP}-hard even to decide whether the optimum objective value is non-zero.
Our reduction also rules out efficient algorithms in graphs of bounded feedback edge number, cutwidth, treewidth + max. degree, or slim tree-cut width, in stark contrast to fixed-parameter tractability (FPT) results obtained for the standard \ECC{} problem by~\citet{kellerhals2023parameterized}.
On the positive side, we can give an algorithm parameterized by the total number $t$ of unsatisfied edges.
%
\begin{restatable}{theorem}{cfminECCFPT}\label{thm:cfminecc-FPT}
    \cfminECC{} is FPT with respect to the total number $t$ of unsatisfied edges.
\end{restatable}
\begin{proof}
    We give a branching algorithm.
    Given an instance $(H = (V, E), \tau)$ of \cfminECC{}, a \emph{conflict} is a triple $(v, e_1, e_2)$ consisting of a single vertex $v$ and a pair of distinctly colored hyperedges $e_1, e_2$ which both contain $v$.
    If $H$ contains no conflicts, then it is possible to satisfy every edge.
    Otherwise, we identify a conflict in $O(r|E|)$ time by scanning the set of hyperedges incident on each node.
    Once a conflict $(v, e_1, e_2)$ has been found, we branch on the two possible ways to resolve this conflict: deleting $e_1$ or deleting $e_2$.
    Here, deleting a hyperedge has the same effect as ``marking'' it as unsatisfied and no longer considering it for the duration of the algorithm.
    We note that it is simple to check in constant time whether a possible branch violates the constraint given by $\tau$; these branches can be pruned.
    Because each branch increases the number of unsatisfied hyperedges by 1, the search tree has depth at most $\edgedeletions$.
    Thus, by computing the search tree in level-order, the algorithm runs in time $O(2^{\edgedeletions}r|E|)$.
\end{proof}
One question left open by the reduction of~\Cref{thm:cfminecc-NPhard} is whether \cfminECC{} is hard when the number~$k$ of colors is bounded.
The standard~\ECC{} objective is in \cclass{P} when $k \leq 2$ and \cclass{NP}-hard whenever $k \geq 3$~\cite{amburg2020clustering}.
We resolve this question via an intermediate hardness result for a \textsc{Vertex Cover} variant which may be of independent interest.
Specifically, we show hardness for the following:

\problembox{\fairbipartiteVC{}}{A bipartite graph $G = (V = A \uplus B, E)$ and an integer $\alpha$.}{Does there exist a vertex cover $C$ of $G$ with $\max\{|C \cap A|, |C \cap B|\} \leq \alpha$?}
\begin{restatable}{theorem}{fairbipartitevertexcoverNPhard}\label{thm:fairbipartiteVC-NPhard}
    \fairbipartiteVC{} is \cclass{NP}-hard.
\end{restatable}
\begin{proof}
    We reduce from \constrainedbipartiteVC{}, for which the input is a bipartite graph $G = (A \uplus B, E)$ along with two integers $\alpha_a$, $\alpha_b$, and the question is whether there exists a vertex cover $C$ with $|C \cap A| \leq \alpha_a$ and $|C \cap B| \leq \alpha_b$.
    This problem was shown to be \cclass{NP}-complete by~\citet{kuo1987efficient}.

    Given an instance $(G = (A \uplus B, E), \alpha_a, \alpha_b)$ of \constrainedbipartiteVC{}, we construct an instance $(G' = (A' \uplus B', E'), \alpha)$ of \fairbipartiteVC{} as follows.
    First, we observe that if $\alpha_a = \alpha_b$ then $(G, \alpha_a, \alpha_b)$ is already an instance of \fairbipartiteVC{}, in which case there is nothing to do.
    We therefore set $\beta = \alpha_b - \alpha_a$ and assume without loss of generality that $\beta > 0$.
    We begin by copying $G'$. That is, for each vertex $a \in A$ we create a vertex $a' \in A'$, for each vertex $b \in B$ we create a vertex $b' \in B'$, and for each edge $ab \in E$ we create an edge $a'b' \in E'$. We call
    the vertices (and edges) that we have created thus far \emph{original} vertices (and edges).
    Next, we create $\beta$ auxiliary vertices $\{a_i \ | \ 1 \leq i \leq \beta \} \subset A'$, and $\beta(\alpha_b + 1)$ auxiliary vertices $\{ b_{ij} \ | \ 1 \leq i \leq \beta, 1 \leq j \leq \alpha_b + 1 \} \subset B'$.
    We also add auxiliary edges $\{a_ib_{ij} \ | \ 1 \leq i \leq \beta, 1 \leq j \leq \alpha_b + 1 \} \subset E'$. Finally, we set $\alpha = \alpha_b$. This concludes the construction.

    It remains to show that the reduction is correct. For the first direction, assume that $C \subseteq V$ is a vertex cover of $G$ with $|C \cap A| \leq \alpha_a$ and $|C \cap B| \leq \alpha_b$.
    We construct a vertex cover $C'$ of $G'$ as follows. For each vertex $a \in C \cap A$, we add $a'$ to $C$, and for each vertex $b \in C \cap B$, we add $b'$ to $C$. We also add the auxiliary vertex $a_i$ to $C$, for each $1 \leq i \leq \beta$.
    Because $C$ is a vertex cover of $E$, $C'$ covers every original edge in $E'$. Because every auxiliary edge in $E'$ is incident on some auxiliary vertex $a_i$, $C'$ also covers all auxiliary edges in $E'$.
    Hence, $C'$ is a vertex cover of $G'$. Moreover, by construction $|C' \cap B'| \leq \alpha_b = \alpha$, and $|C' \cap A'| \leq \alpha_a + \beta = \alpha$.

    For the other direction, assume that there exists a vertex cover $C'$ of $G'$ with $|C' \cap A'|, |C' \cap B'| \leq \alpha$.
    We will construct a vertex cover $C$ of $G$.
    We begin by observing that $C'$ contains every auxiliary vertex $a_i$, for $1 \leq i \leq \beta$.
    Suppose otherwise, i.e., that for some fixed $i$ we have $a_i \notin C'$. Then each of the $\alpha_b + 1$ auxiliary vertices $b_{ij}$, for $1 \leq j \leq \alpha_b + 1$, is contained in $C'$.
    Since $\alpha = \alpha_b$, this contradicts that $|C' \cap B'| \leq \alpha$. So, we conclude that $a_i \in C'$ for each $1 \leq i \leq \beta$.
    We now also assume that $C' \cap B'$ consists entirely of original vertices, since the vertex set resulting from the removal of an auxiliary vertex $b_{ij}$ is still a cover, as it contains $a_i$.
    Now, for each $b' \in C'$, we add $b$ to $C$, and for each original $a' \in C'$, we add $a$ to $C$. Because the auxiliary vertices of $G'$ are not incident to any original edges, the original vertices of $C'$ must cover
    all original edges. Hence, $C$ is a vertex cover of $G$. Since $|C' \cap A'| \leq \alpha = \alpha_a + \beta$ and $C' \cap A'$ contains $\beta$ auxiliary vertices, $C' \cap A'$ must contain no more than $\alpha_a$ original vertices.
    Similarly, $C' \cap B'$ contains at most $\alpha = \alpha_b$ original vertices.
    These properties allow us to conclude that $|C \cap A| \leq \alpha_a$ and $|C \cap B| \leq \alpha_b$, as desired.
\end{proof}
\Cref{thm:fairbipartiteVC-NPhard} is useful to us because we can now reduce from \fairbipartiteVC{} to show hardness for \cfminECC{} and \cfmaxECC{} even when the number $k$ of colors is bounded.

\begin{restatable}{theorem}{cfmineccNPhardboundedk}\label{thm:cfminecc-NPhard-k2}
    \cfminECC{} and \cfmaxECC{} are \cclass{NP}-hard even in 2-regular hypergraphs with $k = 2$.
\end{restatable}
\begin{proof}
    Given an instance $(G = (A \uplus B, E), \alpha)$ of \fairbipartiteVC{}, we construct an instance
    $(H = (V, E_a \uplus E_b), \tau = \alpha)$ of \cfminECC{} with two colors ($c_A$ and $c_B$) as follows.
    For each edge $ab \in E$, we create a vertex $v_{ab} \in V$.
    We call this vertex a \emph{conflict vertex}, and we say that it is associated with $a \in A$ and with $b \in B$.
    Next, for each vertex $a \in A$, we create a hyperedge $e_a$ which contains every conflict vertex associated with $a$. We color this edge $c_A$.
    Similarly, for each vertex $b \in B$ we create a hyperedge $e_b$ of color $c_B$ which contains every conflict vertex associated with $b$.
    Observe that every hyperedge is nonempty, as we may assume that $G$ has no isolated vertices.
    Observe also that every vertex has degree~$2$.
    We set $\tau = \alpha$.

    To see that the reduction is correct, consider first a vertex cover $C$ of $G$ with $|C \cap A|, |C \cap B| \leq \alpha$.
    For each vertex $a \in A \setminus C$, we assign color $c_A$ to every conflict vertex associated with $a$.
    Similarly, for every $b \in B \setminus C$, we assign $c_B$ to every conflict vertex associated with $b$.
    We color remaining vertices arbitrarily. We claim that this is a valid coloring, i.e., every vertex has received exactly one color.
    To see this, observe that every vertex $v_{ab}$ in $V$ is a conflict vertex associated with exactly two vertices $a, b$ in $A \uplus B$.
    Since $C$ is a vertex cover, at least one of $a, b$ is in $C$. Hence, $v_{ab}$ is assigned exactly one color. We next claim that
    our assignment of colors leaves at most $\tau = \alpha$ hyperedges of any single color unsatisfied. Due to our coloring scheme,
    if a hyperedge $e_a$ (resp. $e_b$) associated with a vertex $a \in A$ (resp. $b \in B$) is unsatisfied, then
    $a$ (resp. $b$) is contained in $C$. Since $|C \cap A|, |C \cap B| \leq \alpha$, we conclude that at most $\tau = \alpha$ hyperedges
    of color $c_A$ (resp. $c_B$) are unsatisfied.

    For the other direction, suppose that we have a coloring $\lambda\colon V \rightarrow \{c_A, c_B\}$ which leaves at most $\tau = \alpha$ hyperedges
    of any single color unsatisfied. We construct a vertex cover $C$ of $G$ which contains $a \in A$ (resp. $b \in B$) if and only if the hyperedge associated with $a$ (resp. $b$) is unsatisfied by $\lambda$.
    It is immediate that $|C \cap A|, |C \cap B| \leq \alpha$.
    Now we claim that $C$ is a vertex cover of $G$.
    Consider any edge $ab \in E$. The conflict vertex $v_{ab}$ is contained in both $e_a$, which has color $c_A$, and $e_b$, which has color $c_B$.
    Thus, at least one of $e_a, e_b$ is unsatisfied by $\lambda$, and so $C$ contains at least one of $a,b$.
    Then $C$ is a vertex cover, as desired.

    A simple adjustment to our construction yields the claimed hardness result for \cfmaxECC{}.
    Let $\beta = |A| - |B|$, and assume without loss of generality that $\beta \geq 0$.
    If $\beta = 0$, then there are already an equal number of edges of colors $c_A$ and $c_B$, so we make no adjustment to the construction.
    If $\beta \geq 3$, we add an additional $\beta$ vertices $v_1, v_2, \ldots v_\beta$,
    and an additional $\beta$ hyperedges $e_1 = \{v_1, v_2\}, e_2 = \{v_2, v_3\}, \ldots, e_\beta = \{v_\beta, v_1\}$, constructing a cycle with $\beta$ vertices and edges (a $C_\beta$).
    We assign color $c_B$ to each of these hyperedges. We now observe that in the constructed hypergraph, there are an equal number of
    hyperedges of colors $c_A$ and $c_B$, and every vertex has degree~$2$.
    Finally, if $\beta \in \{1, 2\}$, we add $\beta + 6$ vertices $v_1, v_2 \ldots v_{\beta + 6}$. With three of these vertices we form a triangle
    with each edge having color $c_A$, and with the remaining $\beta + 3$ vertices we form a $C_{\beta+3}$ with each edge having color $c_B$.
    Once again, we observe that there are now an equal number of edges of colors $c_A$ and $c_B$ in the constructed hypergraph, and every vertex has degree~$2$.
    We set $\tau = |A| - \alpha$ if $\beta \geq 3$, or $\tau = |A| + 3 - \alpha$ otherwise.
    Correctness follows from a substantively identical analysis.
\end{proof}
We will now further examine connections between \cfminECC{} and variants of \textsc{Vertex Cover}.
The main idea of all of our reductions is similar to that of~\citet{veldt2023optimal}; we create a vertex for each hyperedge, and add edges for all pairs of distinctly colored overlapping hyperedges.
Recalling~\Cref{sec:generalized-mean}, we call this construction the \emph{conflict graph} associated with an edge-colored hypergraph.

We observe that in addition to the algorithm of~\Cref{thm:pmean-approximation}, a $2$-approximation for \cfminECC{} can also be obtained
by applying a result of~\citet{blum2022sparse} to the constructed conflict graph. Formally, we reduce to the following:

\optproblembox{\SVC{}}{A graph $G = (V = V_1 \cup V_2 \cup \ldots \cup V_k, E)$.}{Find a vertex cover $C \subseteq V$ of $G$ which minimizes $\max_{i \in [k]} \{|C \cap V_i|\}$.}
\begin{restatable}{proposition}{cfmineccreductiontosparsevertexcover}\label{thm:cfminecc-reduce-sparse-vc}
    \cfminECC{} admits a $2$-approximation via reduction to \SVC{}.
\end{restatable}
\begin{proof}
    Let $H = (V, E = E_1 \uplus E_2 \uplus \ldots \uplus E_k)$ be an edge-colored hypergraph with $k$ colors.
    We construct an instance $G = (V' = V_1' \cup V_2' \ldots V_k', E')$ of \SVC{}. For each hyperedge $e$ of color $i$, we create a vertex $v_e$ in $V_i'$.
    Next, for every pair of distinctly colored hyperedges $e_1, e_2$ with $e_1 \cap e_2 \neq \emptyset$ (a \emph{bad hyperedge pair}), we create an edge $v_{e_1}v_{e_2}$.
    This completes the construction of $G$.

    Now, let $\opt_G$ and $\opt_H$ be the optimal objective values for $G$ and $H$, respectively. We claim that
    $\opt_G \leq \opt_H$. Suppose that $S$ is the set of hyperedges unsatisfied by some optimal coloring $\lambda$, so $\max_{i \in [k]} |S \cap E_i| = \opt_H$.
    Because $\lambda$ satisfies all edges in $E \setminus S$, $S$ must contain at least one member of every bad hyperedge pair.
    Hence, the set $C = \{v_e \in V' \ \colon \ e \in S \}$ of vertices in $G$ corresponding to hyperedges in $S$ is a vertex cover in $G$, and has the property that $\max_{i \in [k]} |C \cap V_i'| = \opt_H$.
    It follows that $\opt_G \leq \opt_H$.

    Finally, we show how to lift a solution. Let $C \subseteq V'$ be a vertex cover in $G$ which is 2-approximate; such a cover can be computed in polynomial time~\cite{blum2022sparse}.
    Let $S = \{e \in E : v_e \in C\}$ be the hyperedges of $H$ corresponding to the vertices in $C$.
    Because $C$ is a vertex cover and the edges of $G$ correspond exactly to the bad hyperedge pairs of $H$, it is trivial to compute a
    coloring $\lambda$ which satisfies every edge in $E \setminus S$. The objective value of $\lambda$
    is at worst
    \[
        \max_{i \in [k]} |S \cap E_i| = \max_{i \in [k]} |C \cap V_i'| \leq 2\cdot \opt_G \leq 2\cdot \opt_H.
    \]

\end{proof}

In practice, the algorithm of~\Cref{thm:pmean-approximation} is preferable, as
the conflict graphs associated with edge-colored hypergraphs can be very dense and inefficient to even form explicitly.
%
%

A limitation of both approaches is the reliance on convex optimization as a subroutine.
Hence, an appealing open direction is to develop purely combinatorial algorithms.
Toward this end, we now give a combinatorial $k$-approximation, i.e., one which is constant-factor for every fixed number of edge colors. The algorithm begins by computing a maximal matching $M$ in the conflict graph $G$.
The vertices in $G$ which are unmatched by $M$ are an independent set, so it is possible to simultaneously satisfy all of the corresponding hyperedges.
This algorithm has already been shown to be a $2$-approximation to the standard \minecc{} objective by~\citet{veldt2023optimal}, who also showed that by avoiding explicit construction of $G$, the running time can be made \emph{linear} in the size of the hypergraph, i.e., $O(\sum_{e \in E} |e|)$.
Our contribution is to show that the algorithm also gives a $k$-approximate solution to the color-fair objective; the lower bound is derived from an analysis of the $k$-partite structure of $G$.

\begin{restatable}{theorem}{cfminecccombinatorial}\label{thm-combinatorial-k-approx}
    \cfminECC{} admits a linear-time combinatorial $k$-approximation.
\end{restatable}
\begin{proof}
    We begin by proving that \SVC{} admits a polynomial-time combinatorial $k$-approximation when the vertex classes $V_1, V_2, \ldots, V_k$ are disjoint and every edge has endpoints in distinct vertex classes.
    Let $G = (V = V_1 \uplus V_2 \uplus \ldots \uplus V_k, E)$ be an instance of \SVC{} which satisfies these conditions.
    Let $\opt_G$ be the optimal objective value for \SVC{} on $G$, $M \subseteq E$ be a maximal matching in $G$, and $C \subseteq V$ be $\bigcup_{e \in M} e$.
    We claim that $C$ is a $k$-approximate solution.

    Since $M$ is maximal, $C$ is a vertex cover of $G$. Because every edge has endpoints in distinct vertex classes,
    for each $i \in [k]$ we have that
    \[
        |C \cap V_i| \leq \frac{|C|}{2} \leq |M|.
    \]

    To obtain a lower bound, we observe that because $M$ is a matching, every vertex cover of $G$ has cardinality at least $|M|$.
    It follows from the pigeon-hole principle that $\opt_G \geq |M|/k$. Thus,
    for each $i \in [k]$ we have that $|C \cap V_i| \leq k\cdot\opt_G$, as desired.

    The result for \cfminECC{} now follows from the reduction of~\cref{thm:cfminecc-reduce-sparse-vc}.
\end{proof}

\subsection{Protected-color ECC}\label{sec:protected-class}

\cfminECC{} encodes fairness by ensuring that no single color is heavily ``punished'' (as measured by unsatisfied hyperedges), but what if we are only concerned about a single color that represents a protected interaction type?
It is always trivially possible to ensure that every edge of a given color $c_1$ is satisfied by setting $\lambda(v) = c_1$ for every vertex $v$.
This strategy, however, completely disregards all other edges.
We might reasonably assume that while we are constrained by concern for a given protected color, we still wish to
find a high-quality solution as measured by the traditional clustering objective.
This tension is captured in the following problem definition.

\problembox{\pcECC{}}{An edge-colored hypergraph $H = (V, E, \ell)$, two integers $t, b$, and a protected color $c_1$.}{Is it possible to color $V$ such that at most $t$ edges are unsatisfied, of which at most $b$ have color $c_1$?}

When $t = b$ we recover the standard $\ECC{}$ problem, so \pcECC{} is \cclass{NP}-hard.
However, by rounding fractional solutions to the following linear program we recover bicriteria $(\alpha, \beta)$-approximations. If $\opt_b$ is the minimum possible number of unsatisfied edges subject to the constraint $b$, then these algorithms guarantee at most $\alpha\cdot\opt_b$ unsatisfied edges, of which at most $\beta\cdot b$ have color $c_1$.\looseness=-1
\begin{align}
	\label{lp:pc-ecc}
\begin{aligned}
	\text{min} \quad& \textstyle \sum_{e \in E} \gamma_e \\[4pt]
	\text{s.t.} \quad& \forall v \in V:\\[6pt]
	&\forall c \in [k], e \in E_c:\\[2pt]
    &\textstyle\sum_{e \in E_{c_1}} \gamma_e \leq b \\[1pt]
	&d_v^c, \gamma_e \in [0, 1]
\end{aligned}
\begin{aligned}
	&\\[2pt]
	&\textstyle \sum_{c=1}^k d_v^c \geq k - 1 \\[5pt]
	\quad &d_v^c \leq \gamma_e \quad \forall v \in e \\
    & \\[4pt]
	&\forall c \in [k], v \in V, e \in E
\end{aligned}
\end{align}
All variables have the same meaning as in Program~(\ref{eq:p-mean-ecc}), but we have added the third constraint which encodes that at most $b$ edges of our protected color $c_1$ may be unsatisfied.
Apart from this constraint, LP~(\ref{lp:pc-ecc}) is identical to the canonical \minecc{} LP used to obtain
state-of-the-art approximations~\cite{amburg2020clustering,veldt2023optimal}, and our algorithm can be seen as
an appropriate adaption of those rounding schemes. In the following, $\rho \in (0, \frac{1}{2}]$ is a parameter which allows us to tune the approximation guarantees $\alpha$ and $\beta$.

\begin{restatable}{theorem}{pceccapproximation}\label{thm:pc-ECC-approximation}
    For every $\rho \in (0, \frac{1}{2}]$, there exists a polynomial-time $(\frac{1}{\rho}, \frac{1}{1 - \rho})$-approximation for \pcECC{}.
\end{restatable}
\begin{proof}
    We begin by computing optimal fractional values $\{\gamma_e\}, \{d_v^c\}$ for LP~(\ref{lp:pc-ecc}).

    We claim that for every pair of distinctly colored overlapping hyperedges $e, f$, $\gamma_e + \gamma_f \geq 1$.
    Otherwise, let $v \in e \cap f$, $c_e = \ell(e)$, and $c_f = \ell(f)$. Since $\gamma_e + \gamma_f < 1$,
    it follows from the second constraint of LP~(\ref{lp:pc-ecc}) that $d_v^{c_e} + d_v^{c_f} < 1$. Then
    \[
        \sum_{c = 1}^{k} d_v^c < 1 + k - 2 = k - 1,
    \]
    contradicting the first constraint of LP~(\ref{lp:pc-ecc}).

    Now we show how to color the vertices of our hypergraph $H = (V, E)$ with protected color $E_1$.
    For each hyperedge $e$, we set
    \[
        \hat{\gamma}_e = \begin{cases}
            1 \quad\text{ if } e \in E_1 \text{ and } \gamma_e \geq 1 - \rho \\
            1 \quad\text{ if } e \notin E_1 \text{ and } \gamma_e \geq \rho  \\
            0 \quad\text{ otherwise,}
        \end{cases}
    \]

    and we construct the set $S = \{e \in E \colon \hat{\gamma}_e = 1\}$.

    First we observe that since $\rho \in (0, \frac{1}{2}]$, $\frac{1}{1 - \rho} \leq \frac{1}{\rho}$, and so
    \[
        |S| = \sum_{e \in E} \hat{\gamma}_e \leq \frac{1}{\rho}\sum_{e \in E \setminus E_1} \gamma_e + \frac{1}{1 - \rho}\sum_{e \in E_1} \gamma_e \leq \frac{1}{\rho}\sum_{e \in E} \gamma_e,
    \]
    and
    \[
        |S \cap E_1| \leq \frac{1}{1 - \rho}\sum_{e \in E_1} \gamma_e \leq \frac{b}{1 - \rho}.
    \]

    All that remains is to compute a coloring $\lambda$ which satisfies every hyperedge in $E \setminus S$. If $S$ contains at least one member of every pair of distinctly colored overlapping hyperedges, then computing such a $\lambda$ is trivial.
    We conclude the proof by showing that $S$ has this characteristic.
    Let $e, f \in E$ with $e \cap f \neq \emptyset$ and $\ell(e) \neq \ell(f)$.
    Assume toward a contradiction that $e, f \notin S$.
    Because $e$ and $f$ are distinctly colored, at least one is in $E \setminus E_1$.
    If both $e, f \notin E_1$, then
    \[
        \gamma_e + \gamma_f < 2\rho \leq 1,
    \]
    but we have already shown that $\gamma_e + \gamma_f \geq 1$, so exactly one of $e$ or $f$ must be in $E_1$. Assume without loss of generality that $e \in E_1$ and $f \in E \setminus E_1$.
    Then
    \[
        \gamma_e + \gamma_f < 1 - \rho + \rho = 1,
    \]
    so we have contradicted that $S \cap \{e, f\} = \emptyset$, as desired.

\end{proof}

We also prove the following fixed-parameter tractability result, via a standard branching argument which is essentially identical to that of~\Cref{thm:cfminecc-FPT}.
We defer the details to~\Cref{app:pc}.
\begin{restatable}{theorem}{pceccFPT}\label{thm:pc-ECC-FPT}
    \pcECC{} is FPT with respect to the total number $t$ of unsatisfied edges.
\end{restatable}

Given these positive results, a natural question is whether we may efficiently solve the \minecc{} objective with \emph{multiple} protected colors.
We conclude by showing that as soon as we allow more than one protected color, it becomes hard even to determine whether any solution exists.  

\begin{restatable}{theorem}{multiplepcNPhard}\label{thm:pc-multiple-classes-NP-hard}
	Given $H = (V, E, \ell)$, two integers $b_1, b_2$, and two colors $c_1, c_2$, it is \cclass{NP}-hard to determine whether
	it is possible to color $V$ such that at most $b_1$ edges of color $c_1$ and at most $b_2$ edges of color $c_2$ are unsatisfied. 
\end{restatable}
\begin{proof}
    We once again reduce from \constrainedbipartiteVC{}, for which the input is a bipartite graph $G = (A \uplus B, E)$ along with two integers $\alpha_a$, $\alpha_b$, and the question is whether there exists a vertex cover $C$ with $|C \cap A| \leq \alpha_a$ and $|C \cap B| \leq \alpha_b$.
    We will construct the edge-colored hypergraph for which $G$ is the associated conflict graph.
    Formally, we construct an edge-colored hypergraph $H = (V, E_a \uplus E_b, \ell)$ with two colors $c_a$ and $c_b$ as follows. For each edge $e \in E$, we create a vertex $v_e \in V$.
    For each vertex $a \in A$, we create a hyperedge $e_a \in E_a$ such that $e_a = \{v_e \colon a \in e\}$. Similarly,
    for each vertex $b \in B$, we create a hyperedge $e_b \in E_b$ such that $e_b = \{v_e \colon a \in e\}$.
    For each $e_a \in E_a$ we set $\ell(e_a) = c_a$, and for each $e_b \in E_b$ we set $\ell(e_b) = c_b$.
    We say that our protected colors are $E_a$ and $E_b$, and we set the associated constraints $b_1 = \alpha_a$ and $b_2 = \alpha_b$.

    Thus, if $\lambda$ is a vertex coloring of $H$, $S_\lambda$ is the set of hyperedges unsatisfied by $\lambda$, and we have that $|S \cap E_a| \leq b_1 = \alpha_a$ and $|S \cap E_b| \leq b_2 = \alpha_b$,
    then the set $C = \{a \colon e_a \in C \} \cup \{b \colon e_b \in C \}$ has the properties that $|C \cap A| \leq \alpha_a$ and $|C \cap B| \leq \alpha_b$.
    Moreover, if $ab = e$ is an edge in $e$, then $e_a$ and $e_b$ are distinctly colored, with $e_a \cap e_b = \{v_e\}$.
    Hence, $S$ contains at least one of $e_a, e_b$, meaning that $C$ contains at least one of $a, b$, so $C$ is a vertex cover.

    For the other direction, assume that $C$ is a vertex cover of $G$ with the properties that $|C \cap A| \leq \alpha_a = b_1$ and $|C \cap B| \leq \alpha_b = b_2$,
    so the set $S = \{e_a \colon a \in C\} \cup \{e_b \colon b \in C\}$ has the properties that $|S \cap E_a| \leq b_1$ and $|S \cap E_b| \leq b_2$.
    Observe that, if $e_a$ and $e_b$ are distinctly colored hyperedges which overlap, then $ab \in E$. Thus, at least one of $a, b$ is contained in $C$, and so at least one of $e_a, e_b$ is contained in $S$.
    Then there are no pairs of distinctly colored overlapping hyperedges in $(E_a \uplus E_b)\setminus S$.
    Consequently, it is trivial to compute a coloring $\lambda$ which satisfies every hyperedge not contained in $S$, meaning that it leaves at most $b_1$ hyperedges of color $c_a$ unsatisfied, and at most $b_2$ hyperedges of color $c_b$ unsatisfied.
\end{proof}

\section{Experiments}

\newcommand{\algecc}{\textsf{ST}}
\newcommand{\algcf}{\textsf{CF}}
\newcommand{\algpc}{\textsf{PC}}

\begin{figure*}[t!]
    \centering
    \subfigure[Constraint Violation]{
        \label{fig:pcecc_experiments_constraint}
        \includegraphics[page=5, width=0.35\paperwidth]{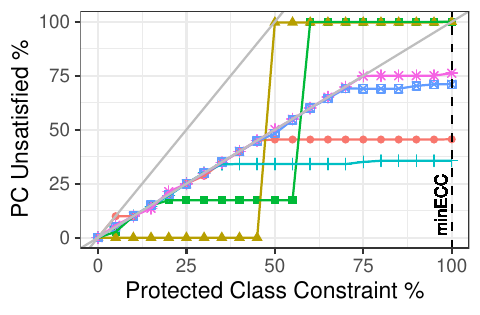}
    }
    ~~~
    \subfigure[Runtime]{
        \label{fig:pcecc_experiments_runtime}
        \includegraphics[page=7, width=0.35\paperwidth]{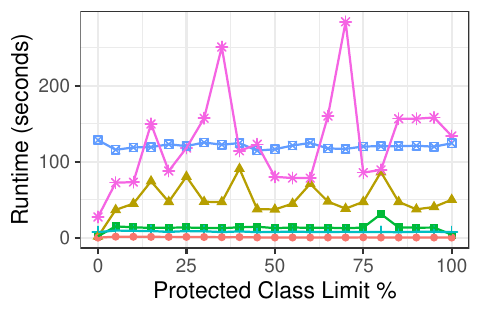}
    }
    ~~~
    \subfigure[Protected Color Objective Approximation Upper Bound]{
        \label{fig:pcecc_experiments_pc_approx}
        \includegraphics[page=6, width=0.35\paperwidth]{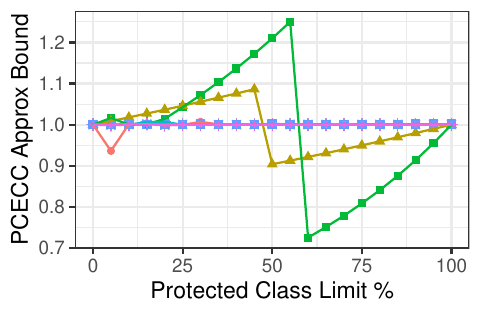}
    }
    ~~~
    \subfigure[MinECC Objective Approximation Upper Bound]{
        \label{fig:pcecc_experiments_st_approx}
        \includegraphics[page=8, width=0.35\paperwidth]{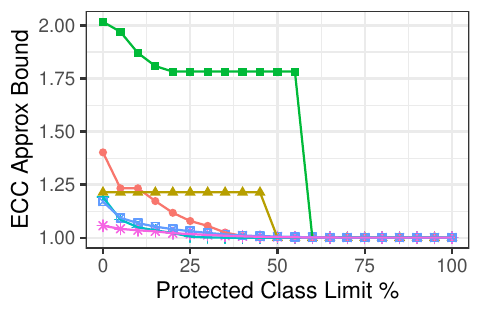}
    }
    $\begin{alignedat}{10}
            \textcolor[RGB]{229, 116, 100}{\bullet} \text{~Brain~}                           & ~~~ &
            \textcolor[RGB]{153, 154, 35}{\blacktriangle} \text{~Cooking~}                   & ~~~ &
            \textcolor[RGB]{73, 178, 117}{\blacksquare} \text{~DAWN~}                        & ~~~ &
            \textcolor[RGB]{67, 164, 241}{\boldsymbol{+}} \text{~MAG-10~}                    & ~~~ &
            \textcolor[RGB]{110, 155, 248}{\boldsymbol{\boxtimes}} \text{~Trivago-Clickout~} & ~~~ &
            \textcolor[RGB]{211, 108, 236}{\boldsymbol{\ast}} \text{~Walmart-Trips~}
        \end{alignedat}$
    \caption{Results running \pcECC{} on multiple datasets with a varying constraint on the number of unsatisfied edges for the color with the median amount of edges.
        Figure \subref{fig:pcecc_experiments_constraint} is the percent of constraint violation of the protected color. Line $y=x$ represents the constraint imposed by the problem definition, and $y=2x$ is the theoretical limit for our bicriteria approximation algorithm. Figure~\subref{fig:pcecc_experiments_runtime} is the runtime in seconds for solving the linear program (note that the rounding step has a negligible runtime). Figure \subref{fig:pcecc_experiments_pc_approx} is the approximation upper bound for the protected color objective. Figure \subref{fig:pcecc_experiments_st_approx} is the approximation upper bound for the standard MinECC objective.
        Approximation upper bounds are given by the objective value for the algorithm divided by the lower bound for the objective determined by the LP relaxation.
    }
    \label{fig:pcecc_experiments}
\end{figure*}

\begin{table}[t]
    \centering
    \caption{We report statistics for a standard suite of ECC datasets from~\cite{amburg2020clustering,veldt2023optimal}, as well as results and runtimes comparing the standard~(\algecc{}) \minecc{} algorithm versus the \cfminECC{}~(\algcf{}). Runtimes are given in seconds taken to solve the LP. Objective numbers are an upper bound on the approximation ratio, given by the objective value for the algorithm divided by the lower bound for the objective determined by the LP relaxation.}
    \label{tab:cfecc_experiments}
    \begin{tabular}{ llllllllllllllll }
        \toprule
        \multirow{2}{*}{\textbf{Dataset}} & \multirow{2}{*}{\textbf{n}} & \multirow{2}{*}{\textbf{m}} & \multirow{2}{*}{\textbf{r}} & \multirow{2}{*}{\textbf{k}} & \multicolumn{2}{c}{\textbf{Runtime (s)}} & \multicolumn{2}{c}{\textbf{ECC Obj}} & \multicolumn{2}{c}{\textbf{CFECC Obj}} \\
        \cmidrule(lr){6-7} \cmidrule(lr){8-9} \cmidrule(lr){10-11}
         & & & & & \multicolumn{1}{c}{\algecc{}} & \multicolumn{1}{c}{\algcf{}} & \multicolumn{1}{c}{\algecc{}} & \multicolumn{1}{c}{\algcf{}} & \multicolumn{1}{c}{\algecc{}} & \multicolumn{1}{c}{\algcf{}} \\
        \cmidrule(lr){1-11}
        Brain            & 638    & 21180  & 2  & 2  & 0.50   & 1.57   & 1.00 & 1.02 & 1.26 & 1.00 \\
        Cooking          & 6714   & 39774  & 65 & 20 & 46.76  & 52.88  & 1.00 & 1.00 & 1.66 & 1.66 \\
        DAWN             & 2109   & 87104  & 22 & 10 & 4.49   & 11.60  & 1.00 & 1.00 & 1.39 & 1.38 \\
        MAG-10           & 80198  & 51889  & 25 & 10 & 8.40   & 15.52  & 1.00 & 1.20 & 1.48 & 1.03 \\
        Trivago-Clickout & 207974 & 247362 & 85 & 55 & 116.94 & 129.55 & 1.00 & 1.37 & 1.68 & 1.01 \\
        Walmart-Trips    & 88837  & 65898  & 25 & 44 & 156.71 & 314.03 & 1.00 & 1.05 & 2.48 & 1.56 \\
        \bottomrule
    \end{tabular}
\end{table}

While we primarily focus on theoretical results, we also implemented and evaluated the performance of our alternate objective ECC algorithms on a standard suite of benchmark hypergraphs~\cite{veldt2023optimal,amburg2020clustering}.
LP-based methods for standard ECC are already known to perform far better in practice than their theoretical guarantees, often yielding optimal solutions~\cite{amburg2020clustering}. Thus, our \maxecc{} algorithms should be viewed as theoretical results that help bridge the theory-practice gap and are thus not the focus of our experiments.

Algorithms based on LP rounding were implemented for the \cfminECC{} and \pcECC{} problems. These algorithms were implemented in the Julia programming language and were run on a research server running Ubuntu 20.04.1 with two AMD EPYC 7543 32-Core Processors and 1 TB of RAM. Our implementation of these algorithms chooses a color that has the lowest LP distance variable among all colors for each node. Note that this can only improve our approximation guarantees over the theoretical approach that only assigns a color if there is an LP variable that is less than $1/2$. We used Gurobi optimization software with default settings to solve the LP for all algorithms, which was the largest bottleneck for running the algorithms. The rounding step has a runtime less than $0.3$ seconds in all cases, which is negligible in comparison with solving the linear program. Statistics about the datasets can be found in Table~\ref{tab:cfecc_experiments}. The code for our experiments can be found at~\url{https://github.com/tommy1019/AltECC}.
We refer to our LP rounding algorithms for \cfminECC{} and \pcECC{} as  \algcf{} and \algpc{} respectively, and refer the standard \minecc{} LP rounding algorithm as \algecc{}.

\textbf{\textsc{Color-fair MinECC}}.
Table \ref{tab:cfecc_experiments} compares \algcf{} against \algecc{}. We first observe that in practice, \algcf{} achieves approximation factors that are much better than the theoretical bound of 2, and in several cases are very close to 1.
Furthermore, \algcf{} even does well approximating the standard \minecc{} objective even though it is not directly designed for it. In the worst case, it achieves a 1.37-approximation, which is only $37\%$ worse than \algecc{}, which finds an optimal solution.
%
By comparison,  \algecc{} is up to $66\%$ worse than \algcf{} for the color-fair objective.
In fact, on average among the datasets we tested, \algcf{} is only $10.66\%$ worse than \algecc{} on the standard \minecc{} objective, while \algecc{} is $32.62\%$ worse than \algcf{} on the color-fair objective.
In other words, the cost of incorporating fairness into the ECC framework is not too high---\algcf{} maintains good performance in terms of the \minecc{} objective while showing significant improvements for the color-fair objective.

\textbf{\textsc{Protected-Color MinECC}}.
Figure \ref{fig:pcecc_experiments} shows the constraint satisfaction, runtime, \algpc{} approximation bounds, and \algecc{} approximation bounds as $b$ (the number of unsatisfied protected edges) varies. Again in we see that our algorithm tends to far exceed theoretical guarantees.
In most cases, Figure \ref{fig:pcecc_experiments_constraint} shows that the number of protected edges left unsatisfied by \algpc{} is below the constraint $b$, even though the algorithm can in theory violate this constraint by up to a factor 2. Results are especially good for Brain, MAG-10, Walmart, and Trivago-Clickout, where the constraint tends to be satisfied and in Figure \ref{fig:pcecc_experiments_pc_approx} we see the objective value for \algcf{} is within a factor $1.007$ of the LP lower bound.
\algpc{} also tends to satisfy the protected color constraint for DAWN and Cooking when $b$ is below $50\%$ of the total number of protected edges, achieving a 1.3-approximation or better in this case. Once $b$ is above this $50\%$ threshold, \algpc{} satisfies no protected edges on DAWN and Cooking. This still matches our theory: if we want to \emph{guarantee} (based on our theory) that we do not leave all protected edges unsatisfied, $b$ must be less than $50\%$ of these edges. We also see in Figure \ref{fig:pcecc_experiments_st_approx} that \algpc{} does a good job of approximating the standard \minecc{} objective (factor 2 or better), while incorporating this protected color constraint.

We note that the rightmost point in the plots in Figure~\ref{fig:pcecc_experiments} corresponds to standard \minecc{} as there is no bound on unsatisfied protected edges. From this, we can examine how badly \algecc{} violates protected color constraints for small $b$ values. On MAG-10 when $b$ is $10\%$ of the number of protected edges, \algecc{} violates the constraint by a factor of $3.60$ whereas \algpc{} only violates the constraint by a factor of $1.01$. This gap is even more pronounced for the DAWN dataset at $10\%$, where \algecc{} violates the constraint by a factor of $10.3$ while \algpc{} satisfies the constraint. Thus, for situations where there is a need to protect certain edge types, standard \minecc{} algorithms do not provide meaningful solutions.

\section{Conclusions}
We have established the first approximation algorithm for hypergraph \maxecc{} and improved the best known approximation factor for graph \maxecc{}. We also introduced two ECC variants that incorporate fairness, designing new approximation algorithms which in practice far exceed their theoretical guarantees. One open direction is to try to further improve the best approximation ratio for \maxecc{}. This may require deviating significantly from previous LP rounding techniques, since it is currently very challenging to improve the approximation by even a small amount using this approach. Another open direction is to improve on the $2^{1/p}$-approximation for \pmeanECC{} when $p< 1$.

\newpage
\section*{Acknowledgements}
This work was supported in part by the Gordon \& Betty Moore Foundation under award GBMF4560 to Blair D. Sullivan, by the National Science Foundation under award IIS-1956286 to Blair D. Sullivan, and by the Army Research Office under award W911NF‐24-1-0156 to Nate Veldt.

\bibliography{refs.bib}

\begin{thebibliography}{23}
\providecommand{\natexlab}[1]{#1}
\providecommand{\url}[1]{\texttt{#1}}
\expandafter\ifx\csname urlstyle\endcsname\relax
  \providecommand{\doi}[1]{doi: #1}\else
  \providecommand{\doi}{doi: \begingroup \urlstyle{rm}\Url}\fi

\bibitem[Ageev and Kononov(2015)]{ageev2015improved}
Alexander Ageev and Alexander Kononov.
\newblock Improved approximations for the max k-colored clustering problem.
\newblock In \emph{Approximation and Online Algorithms: 12th International
  Workshop, WAOA 2014, Wroc{\l}aw, Poland, September 11-12, 2014, Revised
  Selected Papers}, pages 1--10. Springer, 2015.

\bibitem[Ageev and Kononov(2020)]{ageev20200}
Alexander Ageev and Alexander Kononov.
\newblock A 0.3622-approximation algorithm for the maximum k-edge-colored
  clustering problem.
\newblock In \emph{Mathematical Optimization Theory and Operations Research:
  19th International Conference, MOTOR 2020, Novosibirsk, Russia, July 6--10,
  2020, Revised Selected Papers}, pages 3--15. Springer, 2020.

\bibitem[Alhamdan and Kononov(2019)]{alhamdan2019approximability}
Yousef~M Alhamdan and Alexander Kononov.
\newblock Approximability and inapproximability for maximum k-edge-colored
  clustering problem.
\newblock In \emph{Computer Science--Theory and Applications: 14th
  International Computer Science Symposium in Russia, CSR 2019, Novosibirsk,
  Russia, July 1--5, 2019, Proceedings 14}, pages 1--12. Springer, 2019.

\bibitem[Amburg et~al.(2020)Amburg, Veldt, and Benson]{amburg2020clustering}
Ilya Amburg, Nate Veldt, and Austin Benson.
\newblock Clustering in graphs and hypergraphs with categorical edge labels.
\newblock In \emph{Proceedings of The Web Conference 2020}, pages 706--717.
  Association for Computing Machinery, 2020.

\bibitem[Amburg et~al.(2022)Amburg, Veldt, and Benson]{amburg2022diverse}
Ilya Amburg, Nate Veldt, and Austin~R Benson.
\newblock Diverse and experienced group discovery via hypergraph clustering.
\newblock In \emph{Proceedings of the 2022 SIAM International Conference on
  Data Mining (SDM)}, pages 145--153. SIAM, 2022.

\bibitem[Anava et~al.(2015)Anava, Avigdor-Elgrabli, and
  Gamzu]{anava2015improved}
Yael Anava, Noa Avigdor-Elgrabli, and Iftah Gamzu.
\newblock Improved theoretical and practical guarantees for chromatic
  correlation clustering.
\newblock In \emph{Proceedings of the 24th International Conference on World
  Wide Web}, pages 55--65, 2015.

\bibitem[Angel et~al.(2016)Angel, Bampis, Kononov, Paparas, Pountourakis, and
  Zissimopoulos]{angel2016clustering}
Eric Angel, Evripidis Bampis, A~Kononov, Dimitris Paparas, Emmanouil
  Pountourakis, and Vassilis Zissimopoulos.
\newblock Clustering on k-edge-colored graphs.
\newblock \emph{Discrete Applied Mathematics}, 211:\penalty0 15--22, 2016.

\bibitem[Bansal et~al.(2004)Bansal, Blum, and Chawla]{bansal2004correlation}
Nikhil Bansal, Avrim Blum, and Shuchi Chawla.
\newblock Correlation clustering.
\newblock \emph{Machine Learning}, 56:\penalty0 89--113, 2004.

\bibitem[Blum et~al.(2022)Blum, Disser, Feldmann, Gupta, and
  Zych-Pawlewicz]{blum2022sparse}
Johannes Blum, Yann Disser, Andreas~Emil Feldmann, Siddharth Gupta, and Anna
  Zych-Pawlewicz.
\newblock On sparse hitting sets: From fair vertex cover to highway dimension.
\newblock \emph{arXiv preprint arXiv:2208.14132}, 2022.

\bibitem[Bonchi et~al.(2012)Bonchi, Gionis, Gullo, and
  Ukkonen]{bonchi2012chromatic}
Francesco Bonchi, Aristides Gionis, Francesco Gullo, and Antti Ukkonen.
\newblock Chromatic correlation clustering.
\newblock In \emph{Proceedings of the 18th ACM SIGKDD international conference
  on Knowledge discovery and data mining}, pages 1321--1329, 2012.

\bibitem[Bonchi et~al.(2015)Bonchi, Gionis, Gullo, Tsourakakis, and
  Ukkonen]{bonchi2015chromatic}
Francesco Bonchi, Aristides Gionis, Francesco Gullo, Charalampos~E Tsourakakis,
  and Antti Ukkonen.
\newblock Chromatic correlation clustering.
\newblock \emph{ACM Transactions on Knowledge Discovery from Data (TKDD)},
  9\penalty0 (4):\penalty0 1--24, 2015.

\bibitem[Cai and Leung(2018)]{cai2018alternating}
Leizhen Cai and On~Yin Leung.
\newblock Alternating path and coloured clustering.
\newblock \emph{arXiv preprint arXiv:1807.10531}, 2018.

\bibitem[Caton and Haas(2024)]{caton2024fairness}
Simon Caton and Christian Haas.
\newblock Fairness in machine learning: A survey.
\newblock \emph{ACM Computing Surveys}, 56\penalty0 (7):\penalty0 1--38, 2024.

\bibitem[Chierichetti et~al.(2017)Chierichetti, Kumar, Lattanzi, and
  Vassilvitskii]{chierichetti2017fair}
Flavio Chierichetti, Ravi Kumar, Silvio Lattanzi, and Sergei Vassilvitskii.
\newblock Fair clustering through fairlets.
\newblock \emph{Advances in neural information processing systems}, 30, 2017.

\bibitem[Crane et~al.(2024)Crane, Lavallee, Sullivan, and
  Veldt]{crane2024overlapping}
Alex Crane, Brian Lavallee, Blair~D Sullivan, and Nate Veldt.
\newblock Overlapping and robust edge-colored clustering in hypergraphs.
\newblock In \emph{Proceedings of the 2024 conference on Web Search and Data
  Mining}, WSDM '24, 2024.

\bibitem[Kellerhals et~al.(2023)Kellerhals, Koana, Kunz, and
  Niedermeier]{kellerhals2023parameterized}
Leon Kellerhals, Tomohiro Koana, Pascal Kunz, and Rolf Niedermeier.
\newblock Parameterized algorithms for colored clustering.
\newblock In \emph{Proceedings of the AAAI Conference on Artificial
  Intelligence (in press)}, AAAI 2023, 2023.

\bibitem[Klodt et~al.(2021)Klodt, Seifert, Zahn, Casel, Issac, and
  Friedrich]{klodt2021color}
Nicolas Klodt, Lars Seifert, Arthur Zahn, Katrin Casel, Davis Issac, and Tobias
  Friedrich.
\newblock A color-blind 3-approximation for chromatic correlation clustering
  and improved heuristics.
\newblock In \emph{Proceedings of the 27th ACM SIGKDD Conference on Knowledge
  Discovery \& Data Mining}, pages 882--891, 2021.

\bibitem[Kuo and Fuchs(1987)]{kuo1987efficient}
Sy-Yen Kuo and W~Kent Fuchs.
\newblock Efficient spare allocation for reconfigurable arrays.
\newblock \emph{IEEE Design \& Test of Computers}, 4\penalty0 (1):\penalty0
  24--31, 1987.

\bibitem[Lov{\'a}sz(1983)]{lovasz1983submodular}
L{\'a}szl{\'o} Lov{\'a}sz.
\newblock Submodular functions and convexity.
\newblock \emph{Mathematical Programming The State of the Art: Bonn 1982},
  pages 235--257, 1983.

\bibitem[Tovey(1984)]{tovey1984simplified}
Craig~A Tovey.
\newblock A simplified np-complete satisfiability problem.
\newblock \emph{Discrete applied mathematics}, 8\penalty0 (1):\penalty0 85--89,
  1984.

\bibitem[Veldt(2023)]{veldt2023optimal}
Nate Veldt.
\newblock Optimal lp rounding and linear-time approximation algorithms for
  clustering edge-colored hypergraphs.
\newblock In \emph{International Conference on Machine Learning}, pages
  34924--34951. PMLR, 2023.

\bibitem[Xiu et~al.(2022)Xiu, Han, Tang, Cui, and Huang]{xiu2022chromatic}
Qing Xiu, Kai Han, Jing Tang, Shuang Cui, and He~Huang.
\newblock Chromatic correlation clustering, revisited.
\newblock \emph{Advances in Neural Information Processing Systems},
  35:\penalty0 26147--26159, 2022.

\bibitem[Zuckerman(2006)]{zuckerman2006linear}
David Zuckerman.
\newblock Linear degree extractors and the inapproximability of max clique and
  chromatic number.
\newblock In \emph{Proceedings of the thirty-eighth annual ACM symposium on
  Theory of computing}, pages 681--690, 2006.

\end{thebibliography}
\bibliographystyle{plainnat}

\appendix
\onecolumn

\section{\maxecc{} Proofs}
\label{app:maxecc}
We use the following lemma as a small step in our approximation guarantee for hypergraph \maxecc{} (Theorem~\ref{thm:hypermaxecc}). The result was originally used in designing a $1/e^2$-approximation algorithm for graph \maxecc{}~\cite{angel2016clustering}. A full proof can be found in the work of~\cite{angel2016clustering}.
\begin{lemma}[Lemma 4 of~\citet{angel2016clustering}]
	\label{lem:angel}
	Let $\left\{X_1, X_2, \dots, X_j \right\}$ be a set of independent events satisfying $\sum_{i=1}^j \prob[X_i] \leq 1$, then the probability that {at most one} of them happens is greater than or equal to $2/e$.
\end{lemma}

The remainder of our supporting lemmas are new results that we prove for our approximation algorithm for graph \maxecc{}.

\lemdependentynx*
\begin{proof}
	For a node $u$ and color $i$ we use $\overline{X}_u^c$ to denote the event that $u$ does not want $c$.

	We first observe that it is sufficient to show positive correlation between $Y_u^c$ and $N_v$. To see this,
	we use Bayes' Theorem and the independence of $X_u^c$ and $N_v$ to write
	\begin{align*}
		\pgiv{Y_u^c}{N_v \cap X_u^c} & = \frac{\proba{Y_u^c \cap N_v \cap X_u^c}}{\proba{N_v \cap X_u^c}}                   \\
		                             & = \frac{\pgiv{N_v \cap X_u^c}{Y_u^c}\proba{Y_u^c}}{\proba{N_v \cap X_u^c}}           \\
		                             & = \frac{\pgiv{N_v \cap X_u^c}{Y_u^c}\proba{Y_u^c}}{\proba{N_v} \cdot \proba{X_u^c}},
	\end{align*}
	as well as
	\[
		\pgiv{Y_u^c}{X_u^c} = \frac{\proba{X_u^c \cap Y_u^c}}{\proba{X_u^c}} = \frac{\pgiv{X_u^c}{Y_u^c}\proba{Y_u^c}}{\proba{X_u^c}} = \frac{1\cdot \proba{Y_u^c}}{\proba{X_u^c}}.
	\]
	By rearranging terms, we see that our claim is equivalent to
	\[
		\pgiv{N_v \cap X_u^c}{Y_u^c} = \pgiv{N_v}{Y_u^c} \geq \proba{N_v}.
	\]
	We complete the proof by demonstrating the equivalent inequality $\pgiv{Y_u^c}{N_v} \geq \proba{Y_u^c}$. Impose an arbitrary order $c_1, c_2 \ldots, c_k$ on the color set, and without loss of generality assume that $\ell(e) = c = c_k$. In what follows, we always write $c$ for $c_k$ and whenever considering a subscripted color $c_i$ we assume that $i \in W_u \setminus \{c\}$.

	If node $u$ has a strong color, then without loss of generality we assume that that strong color is $c_{k-1}$.
	The remainder of the proof is written for the case where $S_u = \emptyset$, in which case the set of weak colors (other than $c$) that might want node $u$ is $W_u \setminus \{c\} = \{c_1, c_2, \hdots, c_{k-1}\}$. In other words, $W_u$ is the set of colors with indices in $[k-1]$. If instead we had $S_u = \{c_{k-1}\}$ and $W_u \setminus \{c\} = \{c_1, c_2, \hdots, c_{k-2}\}$, the proof works in exactly the same way if we instead considered the set of colors associated with indices in $[k-2]$ instead of indices in $[k-1]$.

	Now, consider all possible subsets of $[k - 1]$ which may define (according to our ordering) those weak colors (besides $c$) which are wanted by $u$. For a particular subset $S \subseteq [k - 1]$,
	we write $X_u^S$ for the event that $u$ wants exactly (not considering $c$) the colors in $S$. In the following, we write $w(S) = \frac{1}{|S| + 1}$ and call this quantity the \emph{weight} of $S$.
	The events $X_u^S$ partition the sample space,
	so we may write
	\begin{align*}
		\proba{Y_u^c} & = \sum_{S \subseteq [k - 1]} \pgiv{Y_u^c}{X_u^S}\cdot\proba{X_u^S} = \proba{X_u^c}\cdot\sum_{S \subseteq [k - 1]} w(S)\cdot\proba{X_u^S}.
	\end{align*}
	In the above, the second equality follows from the observation that, conditioned on $X_u^S$ for any $S \subseteq [k - 1]$, the event $Y_u^c$ occurs if and only if (a) $u$ wants $c$ and (b) $c$ precedes each color in $S$ in the global ordering of colors.
	The conditions (a) and (b) are independent, and occur with probabilities $\proba{X_u^c}$ and $w(S)$, respectively.
	We will now make use of the independence
	of the $X_u^{c_i}$ events to write $\proba{X_u^S}$ as the product of $k - 1$ probabilities.
	\begin{align}\label{eq:yuc}
		\proba{Y_u^c} & = \proba{X_u^c}\cdot\sum_{S \subseteq [k - 1]} \left[ w(S)\cdot \left( \prod_{i \in S}\proba{X_u^{c_i}}\right)\cdot\left(\prod_{j \notin S}\proba{\overline{X}_u^{c_j}}\right) \right].
	\end{align}

	We want to write a similar expression for $\proba{Y_u^c | N_v}$. We begin by using similar reasoning as above to write
	\begin{align*}
		\pgiv{Y_u^c}{N_v} & = \sum_{S \subseteq [k - 1]} \pgiv{Y_u^c}{X_u^S \cap N_v}\cdot\pgiv{X_u^S}{N_v} \\
		                  & = \pgiv{X_u^c}{N_v}\cdot \sum_{S \subseteq [k - 1]} w(S)\cdot\pgiv{X_u^S}{N_v}  \\
		                  & = \proba{X_u^c}\cdot \sum_{S \subseteq [k - 1]} w(S)\cdot\pgiv{X_u^S}{N_v}.
	\end{align*}
	Next, we observe that the $X_u^{c_i}$ events remain independent even when conditioned on $N_v$. This allows us to write

	\begin{align*}
		\pgiv{Y_u^c}{N_v} = \proba{X_u^c}\cdot\sum_{S \subseteq [k - 1]} \left[ w(S)\cdot \left( \prod_{i \in S}\pgiv{X_u^{c_i}}{N_v}\right)\cdot\left(\prod_{j \notin S}\pgiv{\overline{X}_u^{c_j}}{N_v}\right) \right].
	\end{align*}
	To simplify this expression further, we claim that for every $i \in [k - 1]$,
	\[\pgiv{X_u^{c_i}}{N_v} = \pgiv{X_u^{c_i}}{\overline{X}_v^{c_i}}. \]
	This can be derived from the observation
	that $\bigcap_{j \neq i} \overline{X}_v^{c_j}$ is independent of $X_u^{c_i}$ when conditioned on $\overline{X}_v^{c_i}$, as well as being unconditionally independent of $X_v^{c_i}$.
	We then manipulate the definition of conditional probability to see that
	\begin{align*}
		\pgiv{X_u^{c_i}}{N_v} & = \frac{\proba{X_u^{c_i} \cap \bigcap_{j \neq i} \overline{X}_v^{c_j} \cap \overline{X}_v^{c_i}}}{\proba{\bigcap_{j \neq i} \overline{X}_v^{c_j} \cap \overline{X}_v^{c_i}}}        \\
		                      & = \frac{\proba{X_u^{c_i} \cap \bigcap_{j \neq i} \overline{X}_v^{c_j} \cap \overline{X}_v^{c_i}}}{\proba{\bigcap_{j \neq i} \overline{X}_v^{c_j}}\cdot\proba{\overline{X}_v^{c_i}}} \\
		                      & = \frac{\pgiv{X_u^{c_i} \cap \bigcap_{j \neq i} \overline{X}_v^{c_j}}{\overline{X}_v^{c_i}}}{\proba{\bigcap_{j \neq i} \overline{X}_v^{c_j}}}                                       \\
		                      & = \frac{\pgiv{X_u^{c_i}}{\overline{X}_v^{c_i}}\cdot\pgiv{\bigcap_{j \neq i} \overline{X}_v^{c_j}}{\overline{X}_v^{c_i}}}{\proba{\bigcap_{j \neq i} \overline{X}_v^{c_j}}}           \\
		                      & = \frac{\pgiv{X_u^{c_i}}{\overline{X}_v^{c_i}}\cdot\proba{\bigcap_{j \neq i} \overline{X}_v^{c_j}}}{\proba{\bigcap_{j \neq i} \overline{X}_v^{c_j}}}                                \\
		                      & = \pgiv{X_u^{c_i}}{\overline{X}_v^{c_i}}.
	\end{align*}

	We can now write a suitable counterpart to Equation~\eqref{eq:yuc}:
	\begin{align}\label{eq:yuc-conditional-nv}
		\pgiv{Y_u^c}{N_v} = \proba{X_u^c}\cdot\sum_{S \subseteq [k - 1]} \left[ w(S)\cdot \left( \prod_{i \in S}\pgiv{X_u^{c_i}}{\overline{X}_v^{c_i}}\right)\cdot\left(\prod_{j \notin S}\pgiv{\overline{X}_u^{c_j}}{\overline{X}_v^{c_j}}\right) \right].
	\end{align}

	We will complete the proof by showing that the RHS of Equation~\eqref{eq:yuc} is a lower bound for the RHS of Equation~\eqref{eq:yuc-conditional-nv}. We accomplish this in $k - 1$ steps, one for each color besides $c$.
	In the first step, we begin by rearranging the terms of Equation~\eqref{eq:yuc-conditional-nv} to isolate $\pgiv{X_u^{c_1}}{\overline{X}_v^{c_1}}$ and $\pgiv{\overline{X}_u^{c_1}}{\overline{X}_v^{c_1}}$:

	\begin{align*}
		\begin{split}
			\pgiv{Y_u^c}{N_v} = \proba{X_u^c} & \cdot \biggl[ \pgiv{X_u^{c_1}}{\overline{X}_v^{c_1}}\cdot\sum_{S \ni 1} w(S)\cdot \left( \prod_{1 \neq i \in S}\pgiv{X_u^{c_i}}{\overline{X}_v^{c_i}}\right)\cdot\left(\prod_{j \notin S}\pgiv{\overline{X}_u^{c_j}}{\overline{X}_v^{c_j}}\right)                \\
			                                  & + \pgiv{\overline{X}_u^{c_1}}{\overline{X}_v^{c_1}}\cdot \sum_{S \not\ni 1} w(S) \cdot  \left( \prod_{i \in S}\pgiv{X_u^{c_i}}{\overline{X}_v^{c_i}}\right)\cdot\left(\prod_{1 \neq j \notin S}\pgiv{\overline{X}_u^{c_j}}{\overline{X}_v^{c_j}}\right) \biggr].
		\end{split}
	\end{align*}

	In the above, we refer to the first and second sums as the $c_1$-\emph{wanted coefficient} and the $c_1$-\emph{not-wanted coefficient}, respectively.
	Observe that there exists a natural bijection between the summands in these two coefficients. Specifically, we map the summand in the $c_1$-wanted coefficient
	corresponding to set $S$ to the summand in the $c_1$-not-wanted coefficient corresponding to the set $S \setminus \{1\}$. These two summands are identical up to the
	difference between $w(S)$ and $w(S \setminus \{1\})$. By definition, the former weight is smaller. Hence, the $c_1$-wanted coefficient is smaller than
	the $c_1$-not-wanted coefficient.

	Next, we claim that $\pgiv{X_u^{c_1}}{\overline{X}_v^{c_1}} \leq \proba{X_u^{c_1}}$, and (equivalently) $\pgiv{\overline{X}_u^{c_1}}{\overline{X}_v^{c_1}} \geq \proba{\overline{X}_u^{c_1}}$.
	To prove this claim, we consider two cases. If $x_u^{c_1} \leq x_v^{c_1}$, then $\pgiv{X_u^{c_1}}{\overline{X}_v^{c_1}} = 0 \leq x_u^{c_1} = \proba{X_u^{c_1}}$. Otherwise, $x_u^{c_1} > x_v^{c_1}$, and we
	have $\pgiv{X_u^{c_1}}{\overline{X}_v^{c_1}} = \frac{x_u^{c_1} - x_v^{c_1}}{1 - x_v^{c_1}} = x_u^{c_1}\frac{1 - x_v^{c_1}/x_u^{c_1}}{1 - x_v^{c_1}} \leq x_u^{c_1} = \proba{X_u^{c_1}}$.

	We now rewrite our expression for $\pgiv{Y_u^c}{N_v}$ in terms of \emph{unconditional} probabilities of $u$ (not) wanting color $c_1$. In the following, let
	$0 \leq d = \proba{X_u^{c_1}} - \pgiv{X_u^{c_1}}{\overline{X}_v^{c_1}}$.

	\begin{align*}
		\begin{split}
			\pgiv{Y_u^c}{N_v} = \proba{X_u^c} & \cdot \biggl[ (\proba{X_u^{c_1}} - d)\cdot\sum_{S \ni 1} w(S)\cdot \left( \prod_{1 \neq i \in S}\pgiv{X_u^{c_i}}{\overline{X}_v^{c_i}}\right)\cdot\left(\prod_{j \notin S}\pgiv{\overline{X}_u^{c_j}}{\overline{X}_v^{c_j}}\right)                \\
			                                  & + (\proba{\overline{X}_u^{c_1}} + d)\cdot \sum_{S \not\ni 1} w(S) \cdot  \left( \prod_{i \in S}\pgiv{X_u^{c_i}}{\overline{X}_v^{c_i}}\right)\cdot\left(\prod_{1 \neq j \notin S}\pgiv{\overline{X}_u^{c_j}}{\overline{X}_v^{c_j}}\right) \biggr].
		\end{split}
	\end{align*}

	Our observation that the $c_1$-wanted coefficient is smaller than the $c_1$-not-wanted coefficient yields the bound we desire:
	\begin{align*}
		\begin{split}
			\pgiv{Y_u^c}{N_v} \geq \proba{X_u^c} & \cdot \biggl[ \proba{X_u^{c_1}}\cdot\sum_{S \ni 1} w(S)\cdot \left( \prod_{1 \neq i \in S}\pgiv{X_u^{c_i}}{\overline{X}_v^{c_i}}\right)\cdot\left(\prod_{j \notin S}\pgiv{\overline{X}_u^{c_j}}{\overline{X}_v^{c_j}}\right)                \\
			                                     & + \proba{\overline{X}_u^{c_1}}\cdot \sum_{S \not\ni 1} w(S) \cdot  \left( \prod_{i \in S}\pgiv{X_u^{c_i}}{\overline{X}_v^{c_i}}\right)\cdot\left(\prod_{1 \neq j \notin S}\pgiv{\overline{X}_u^{c_j}}{\overline{X}_v^{c_j}}\right) \biggr].
		\end{split}
	\end{align*}

	This completes the first of $k - 1$ steps. In the next step, we begin by rearranging terms in the inequality above to isolate $\pgiv{X_u^{c_2}}{\overline{X}_v^{c_2}}$ and $\pgiv{\overline{X}_u^{c_2}}{\overline{X}_v^{c_2}}$. We then apply the same analysis, yielding another lower bound on $\pgiv{Y_u^c}{N_v}$, this time written
	in terms of unconditional probabilities of $u$ (not) wanting colors $c_1$ and $c_2$, and conditional probabilities of $u$ (not) wanting colors $c_3, c_4, \ldots c_{k-1}$. After $k - 1$ steps, our lower bound becomes identical (up to rearranging terms) to the RHS of Equation~\ref{eq:yuc}, as desired.
\end{proof}

\lemsumtoprod*
\begin{proof}
	Observe that the claim is trivially true when $m = 1$. We proceed via induction on $m$. For $m > 1$, we have $\sum_{t = 1}^{m - 1} \leq \beta - x_m$, so by the inductive hypothesis
	\[ \prod_{t = 1}^m (1 - x_t) = (1 - x_m)\cdot \prod_{t=1}^{m-1}(1 - x_t) \geq (1 - x_m)(1 - \beta + x_m) = 1 - \beta + x_m(\beta - x_m) \geq 1 - \beta.\]
\end{proof}

\lembounding*
\begin{proof}
	Let $\mathcal{S} = \{\vx \in [0,2/3]^m \colon \sum_{i = 1}^m x_i = 1\} \subsetneq \mathcal{D}$ be the set of input vectors for $f$ whose entries sum exactly to 1. Let $\mathcal{I}\subsetneq \mathcal{S}$ be the set of vectors in $\mathcal{S}$ with one entry equal to $2/3$, one entry equal to $1/3$, and all other entries equal to 0. More formally, if $\vx \in \mathcal{I}$, this means there exists two distinct indices $\{i,j\}$ such that $x_i = 2/3$, $x_j = 1/3$, and $x_k = 0$ for every $k \notin \{i,j\}$. Our goal is to prove there exists some $\vx \in \mathcal{I}$ that minimizes $f$. We prove this in two steps: (1) we show that there exists a minimizer of $f$ in $\mathcal{S}$, and then (2) we show how to convert an arbitrary vector $\vx \in \mathcal{S}$ into a new vector $\vx^* \in \mathcal{I}$ satisfying $f(\vx^*) \leq f(\vx)$. Given $\vx^* \in \mathcal{I}$, the value $f(\vx^*)$ in Eq.~\eqref{eq:minval} amounts to a simple function evaluation, realizing that $P(\vx^*,t) = 0$ for $t \geq 3$.

	\paragraph{Step 1: Proving minimizers exist in $\mathcal{S}$.}

	Let $\vx \in \mathcal{D} \setminus \mathcal{S}$, and let $y \in  [m]$ be an entry such that $x_y < 2/3$. Set 
	\begin{equation*}
	\delta = \min \left\{ 1 - \sum_{i = 1}^m x_i, \frac{2}{3}- x_y\right\} > 0
	\end{equation*} and define a new vector $\vx' \in \mathcal{D}$ by
	\begin{equation}
		x_i' = \vx'(i) = \begin{cases}
			x_i          & \text{ if $i \neq y$} \\
			x_i + \delta & \text{otherwise}.
		\end{cases}
	\end{equation}
	By our choice of $\delta$ we know that $\vx' \in \mathcal{D}$, and we will prove this satisfies $f(\vx') \leq f(\vx)$. To do so, we re-write $f(\vx)$ in a way that isolates terms involving $x_y.$ Let $\noty{} = [m] \setminus \{y\}$. Recall that $f$ is a linear combination of terms $P(\vx,t)$, where $P(\vx,t)$ represents the probability that exactly $t$ events out of a set of $m$ independent events occur. Entry $x_i$ of $\vx$ is the probability that $i$th event occurs. We can therefore re-write:
	\begin{align}
		\label{eq:pxtwithy}
		P(\vx,t) & = P(\vx[\noty{}],t) (1-x_y) + P(\vx[\noty{}], t-1) x_y.
	\end{align}
	To explain this in more detail, we use a slight abuse of terminology and refer to $[m]$ as a set of events. The first term in Eq.~\eqref{eq:pxtwithy} is the probability that exactly $t$ events in $\noty{}$ occur and that event $y$ does not occur. The second term is the probability that $y$ does occur and exactly $t -1$ events in $\noty{}$ occur. We therefore re-write $f(\vx)$ as
	\begin{align}
		f(\vx) & = \sum_{t=0}^m a_t P(\vx, t)                                                                       & (\text{definition of $f$}) \notag            \\
		       & = \sum_{t=0}^m a_t \left(P(\vx[\noty{}],t) (1-x_y) + P(\vx[\noty{}], t-1) x_y \right)              & (\text{from Eq.~\eqref{eq:pxtwithy}}) \notag \\
		       & = \sum_{t=0}^m a_t P(\vx[\noty{}],t) (1-x_y) + \sum_{t=0}^m a_t P(\vx[\noty{}], t-1) x_y           & (\text{separating terms}) \notag             \\
		       & = \sum_{t=0}^m a_t P(\vx[\noty{}],t) (1-x_y) + \sum_{t=-1}^{m-1} a_{t+1} P(\vx[\noty{}], t) x_y    & (\text{change of variables}) \notag          \\
		       & = \sum_{t=0}^{m-1} a_t P(\vx[\noty{}],t) (1-x_y) + \sum_{t=0}^{m-1} a_{t+1} P(\vx[\noty{}], t) x_y & \label{eq:fvx_p_0}                           \\
		       & = \sum_{t=0}^{m-1} P(\vx[\noty{}],t) \left( a_t(1-x_y) + a_{t+1}x_y \right)                        & \label{eq:last}
	\end{align}
	In Eq.~\eqref{eq:fvx_p_0} we have used the fact that $P(\vx[\noty{}],\ell) = 0$ if $\ell < 0$ or $\ell > m-1$.
	Observe that if we replace the value of $x_y$ with $x_y + \delta$ in Eq.~\eqref{eq:last} this gives the value of $f(\vx')$. We can then see that
	\begin{align*}
		f(\vx') & = \sum_{t=0}^{m-1} P(\vx[\noty{}],t) \left( a_t(1-x_y - \delta) + a_{t+1}(x_y + \delta) \right)       \\
		        & = \sum_{t=0}^{m-1} P(\vx[\noty{}],t) \left( a_t(1-x_y) + a_{t+1}x_y + \delta (a_{t+1} - a_t)  \right) \\
		        & = f(\vx) + \delta \sum_{t=0}^{m-1} (a_{t+1} - a_t).
	\end{align*}
	Using the constraint $a_{t+1} \leq a_t$ we notice that the last term is always less that or equal to zero. Hence, we have $f(\vx') \leq f(\vx)$. Recall that $\delta = \min\{ 1 - \sum_{i =1}^m x_i, 2/3 - x_y\}$. If $\delta = 1 - \sum_{i=1}^m x_i$, then we can see that $\vx' \in \mathcal{S}$ and we are done. Otherwise, $\delta = 2/3- x_y$, and $x_y' = 2/3$. We can then apply the same exact procedure again, by increasing the value of some index $x_z$ (where $z \neq y$), which we know satisfies $x_z \leq 1/3$. The second application of this procedure is guaranteed to produce a vector in $\mathcal{S}$.

	\paragraph{Step 2: Editing a vector from $\mathcal{S}$ to $\mathcal{I}$.}
	The proof structure for Step 2 is very similar as Step 1, but is more involved as it requires working with two entries of $\vx$ rather than one. Let $\vx \in \mathcal{S} \setminus \mathcal{I}$, which means we can identify two entries $x_y$ and $x_z$ in the vector (with $y \neq z$) satisfying the inequality
	\begin{equation}
		\label{eq:yzinequality}
		0 < x_y \leq x_z < 2/3.
	\end{equation}
	Given these entries, set $\delta = \min\{x_y, 2/3 - x_z \}$ and define a new vector $\vx'$ by
	\begin{equation}
		x_i' = \vx'(i) = \begin{cases}
			x_i          & \text{ if $i \notin \{y,z\}$} \\
			x_i - \delta & \text{ if $i = y$}            \\
			x_i + \delta & \text{ if $i = z$.}
		\end{cases}
	\end{equation}
	Observe that $\vx' \in \mathcal{S}$ and that either $x_y' = 0$ or $x_z' = 2/3$. We will show that $f(\vx') \leq f(\vx)$, by re-writing $f(\vx)$ in a way that isolates terms involving $x_y$ and $x_z$. Let $\notyz{} = [m] \backslash \{y,z\}$. Similar to the proof of Step 1, we can re-write $P(\vx,t)$ as
	\begin{align}
		P(\vx,t) & = P(\vx[\notyz{}],t) (1-x_y)(1-x_z) \label{eq:pxt1}                               \\
		         & \;\;+ P(\vx[\notyz{}], t-1) \left[x_y (1-x_z) + (1-x_y)x_z\right] \label{eq:pxt2} \\
		         & \;\; + P(\vx[\notyz{}], t-2) x_y x_z. \label{eq:pxt3}
	\end{align}
	Line~\eqref{eq:pxt1} captures the probability that exactly $t$ events from the set $\notyz{}$ happen, and that neither of the events $y$ or $z$ happen. Line~\eqref{eq:pxt2} is the probability that $t-1$ events from $\notyz{}$ happen, and exactly one of the events $\{y,z\}$ happens. Finally, line~\eqref{eq:pxt3} is the probability that both $y$ and $z$ happen, and $t-2$ of the events in $\notyz{}$ happen. For simplicity we now define
	\begin{align*}
		\alpha & = (1-x_y)(1-x_z)           \\
		\beta  & = x_y (1-x_z) + (1-x_y)x_z \\
		\gamma & = x_y x_z,
	\end{align*}
	so that we can re-write $f(\vx)$ as follows:
	\begin{align}
		f(\vx) & = \sum_{t=0}^m a_t P(\vx, t)                                                                                                                                & (\text{def. of $f$}) \notag      \\
		       & = \sum_{t=0}^m a_t \left(P(\vx[\notyz{}],t)\alpha + P(\vx[\notyz{}], t-1) \beta + P(\vx[\notyz{}], t-2) \gamma  \right)                                     & (\text{exp. $P(\vx,t)$}) \notag  \\
		       & = \sum_{t=0}^m a_t P(\vx[\notyz{}],t)\alpha + \sum_{t=0}^m a_t P(\vx[\notyz{}], t-1) \beta +  \sum_{t=0}^m a_t P(\vx[\notyz{}], t-2) \gamma                 & (\text{sep. sums}) \notag        \\
		       & = \sum_{i=0}^m a_i P(\vx[\notyz{}],i)\alpha + \sum_{j=-1}^{m-1} a_{j+1} P(\vx[\notyz{}], j) \beta +  \sum_{k=-2}^{m-2} a_{k+2} P(\vx[\notyz{}], k) \gamma   & (\text{redef. variables}) \notag \\
		       & = \sum_{i=0}^m a_i P(\vx[\notyz{}],i)\alpha + \sum_{j=0}^{m-1} a_{j+1} P(\vx[\notyz{}], j) \beta +  \sum_{k=0}^{m-2} a_{k+2} P(\vx[\notyz{}], k) \gamma     & \label{eq:littlel}               \\
		       & = \sum_{i=0}^{m-2} a_i P(\vx[\notyz{}],i)\alpha + \sum_{j=0}^{m-2} a_{j+1} P(\vx[\notyz{}], j) \beta +  \sum_{k=0}^{m-2} a_{k+2} P(\vx[\notyz{}], k) \gamma & \label{eq:bigl}                  \\
		       & = \sum_{t=0}^{m-2} P(\vx[\notyz{}],t) \left( a_t \alpha +  a_{t+1}  \beta +   a_{t+2}  \gamma \right).                                                       & \label{eq:rewritef}
	\end{align}
	In Eq.~\eqref{eq:littlel} and Eq.~\eqref{eq:bigl} we have used the fact that $P(\vx[\notyz{}],\ell) = 0$ if $\ell < 0$ or $\ell > m-2$. Consider now the vector $\vx'$ obtained by perturbing entries $x_y$ and $x_z$, and define
	\begin{equation*}
		\begin{array}{llll}
			\alpha' & = (1-x_y')(1-x_z')                                               \\
			        & = (1-(x_y-\delta))(1-(x_z +\delta))                              \\
			        & = \alpha + \delta(x_y-x_z) - \delta^2                            \\                                                                                                           \\
			\beta'  & = x_y' (1-x_z') + (1-x_y')x_z'                                   \\
			        & = (x_y-\delta) (1-(x_z +\delta)) + (1-(x_y-\delta))(x_z +\delta) \\
			        & = \beta + 2\delta (x_z - x_y) +2  \delta^2                       \\                                                                     \\
			\gamma' & = x_y' x_z'                                                      \\
			        & = (x_y-\delta) (x_z +\delta)                                     \\
			        & = \gamma + \delta (x_y - x_z) - \delta^2.
		\end{array}
	\end{equation*}
	In the expression for $f(\vx)$ given in Eq.~\eqref{eq:rewritef}, the entries $x_y$ and $x_z$ appear only in the terms $\alpha$, $\beta$, and $\gamma$. Therefore, if we replace $\{\alpha, \beta, \gamma\}$ with $\{\alpha', \beta', \gamma'\}$, we obtain a similar expression for $f(\vx')$. We can then see that

	\begin{align*}
		f(\vx') - f(\vx) & = \sum_{t=0}^{m-2} P(\vx[\notyz{}],t) \left( a_t (\alpha'-\alpha) +  a_{t+1}  (\beta'-\beta) +   a_{t+2}  (\gamma'-\gamma) \right).
	\end{align*}

	Isolating the expression inside the summation we can see

	\begin{align*}
		   & a_t (\alpha'-\alpha) +  a_{t+1}  (\beta'-\beta) +   a_{t+2}  (\gamma'-\gamma)                                             \\
		=~ & a_t (\delta(x_y-x_z) - \delta^2) + a_{t+1}  (2\delta (x_z - x_y) +2  \delta^2) + a_{t+2}  (\delta (x_y - x_z) - \delta^2) \\
		=~ & a_t (\delta(x_y-x_z) - \delta^2) -  2a_{t+1}  (\delta (x_y - x_z) -  \delta^2) + a_{t+2}  (\delta (x_y - x_z) - \delta^2) \\
		=~ & (a_t - 2a_{t+1} + a_{t+2}) (\delta(x_y - x_z) - \delta^2).
	\end{align*}

	Plugging this back into the equation for $f(\vx') - f(\vx)$ and using the fact that $a_t - 2a_{t+1} + a_{t+2} \geq 0$ and $(x_y - x_z) - \delta < 0$, we have

	\[
		f(\vx') - f(\vx) = \sum_{t=0}^{m-2} P(\vx[\notyz{}],t) \left( (a_t - 2a_{t+1} + a_{t+2}) (\delta(x_y - x_z) - \delta^2)\right)\leq 0.
	\]

	If we start with an arbitrary vector $\vx \in \mathcal{S}$, applying one iteration of the above procedure will ensure that \emph{at least} one entry of $\vx$ will move to one of the endpoints $\{0,2/3\}$ of the interval $[0,2/3]$. Therefore, applying this procedure $m$ (or fewer) times to an arbitrary vector $\vx \in \mathcal{S}$ will produce a vector $\vx^* \in \mathcal{I}$ that satisfies $f(\vx^*) \leq f(\vx)$.
\end{proof}

\lemboundingconstraints*
\begin{proof}
	The first inequality (monotonicity) is easy to check for both sequences. The most involved step is checking that the second inequality holds for the sequence in~\eqref{eq:complicatedat}. To show it holds, we note that the inequality $a_t +  a_{t+2} \geq 2a_{t+1}$ effectively corresponds to a discrete version of convexity. We first extend the definition of the sequence in~\eqref{eq:complicatedat} to all positive reals by defining the function
	\begin{equation*}
		h(x) = \frac{2}{9}\left( \frac{1}{x+1} + \frac{1}{x+3} \right) + \frac{5}{9} \left( \frac{1}{x+2} \right),
	\end{equation*}
	which we can easily prove is convex over the interval $[0,\infty)$ by noting that its second derivative is
	\[
		h''(x) = \frac{1}{9}\left( \frac{10}{(x+2)^3} + \frac{4}{(x+3)^3} + \frac{4}{(x+1)^3} \right),
	\]
	which is greater than $0$ for all $x \geq 0$. From the definition of convexity for any $\beta \in [0, 1]$ and any $x, y \geq 0$ we have
	\[
		h(x\beta + (1-\beta) y) \leq \beta h(x) + (1-\beta) h(y).
	\]
	Now if we consider the specific values $x=t$, $y=t+2$, and $\beta=1/2$ we get
	\begin{align*}
		h\left(t\frac{1}{2} + \frac{1}{2} (t+2)\right) & \leq \frac{1}{2} h(t) + \frac{1}{2} h(t+2) \\
		\implies  h(t+1)                               & \leq \frac{1}{2} h(t) + \frac{1}{2} h(t+2) \\
		\implies  2 h(t+1)                             & \leq h(t) + h(t+2)                         \\
		\implies 2a_{t+1}                              & \leq a_t + a_{t+2}.
	\end{align*}
	One can use a similar approach to show (even more easily) that the sequence $a_t = 1/(1+g+t)$ also satisfies the second inequality.
\end{proof}

\section{Full proof of~\texorpdfstring{\Cref{thm:cfminecc-NPhard}}{}}
\label{appendix:color-fair}
\label{app:cf}
\cfmineccNPhard*
{
\usetikzlibrary{math}

\definecolor{c1}{RGB}{0, 114, 178}
\definecolor{c2}{RGB}{213, 94, 0}
\definecolor{c3}{RGB}{0, 158, 115}
\definecolor{c3prime}{RGB}{240, 228, 66}
\definecolor{b1}{RGB}{230, 159, 0}
\definecolor{b2}{RGB}{86, 180, 233}
\definecolor{b3}{RGB}{204, 121, 167}


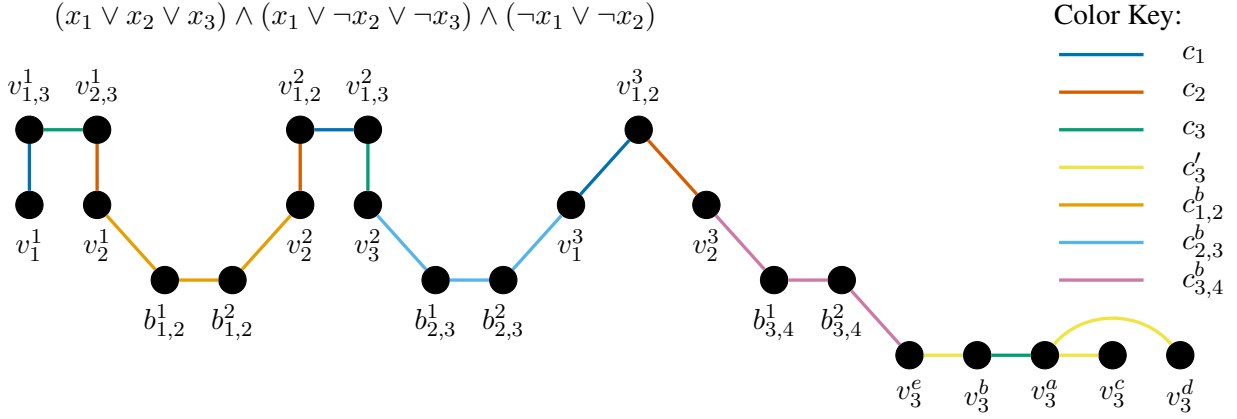
\begin{figure}[t]
    \centering
    \begin{tikzpicture}

        \tikzmath{
            \leftx = 0;
            \dy = 1;
            \confy = 0;
            \freey = \confy - \dy;
            \bridgey = \confy - 2*\dy;
            \sparey = \confy - 3*\dy;
            \dx = .9;
            \edgewidth = 1.25;
            \keyx = \leftx + 15*\dx;
            \keyy = \confy + 1.5;
            \keydx = 1.5;
            \keydy = -.5;
            \keytitlex = \keyx-.2;}

        \node[circle, label=right:$(x_1 \lor x_2 \lor x_3) \land (x_1 \lor \neg x_2 \lor \neg x_3) \land (\neg x_1 \lor \neg x_2)$] at (\leftx, \keyy) (formula) {};

        \node[circle, draw, fill=black, label=below:$v_1^1$] at (\leftx, \freey) (free11) {};
        \node[circle, draw, fill=black, label=above:$v_{1,3}^1$] at (\leftx, \confy) (conf131) {};
        \node[circle, draw, fill=black, label=above:$v_{2,3}^1$] at (\leftx + \dx, \confy) (conf231) {};
        \node[circle, draw, fill=black, label=below:$v_2^1$] at (\leftx + \dx, \freey) (free21) {};

        \draw[c1, line width=\edgewidth pt] (free11) -- (conf131);
        \draw[c3, line width=\edgewidth pt] (conf131) -- (conf231);
        \draw[c2, line width=\edgewidth pt] (conf231) -- (free21);

        \node[circle, draw, fill=black, label=below:$b_{1,2}^1$] at (\leftx + 2*\dx, \bridgey) (bridge121) {};
        \node[circle, draw, fill=black, label=below:$b_{1,2}^2$] at (\leftx + 3*\dx, \bridgey) (bridge122) {};

        \node[circle, draw, fill=black, label=below:$v_2^2$] at (\leftx + 4*\dx, \freey) (free22) {};
        \node[circle, draw, fill=black, label=above:$v_{1,2}^2$] at (\leftx + 4*\dx, \confy) (conf122) {};
        \node[circle, draw, fill=black, label=above:$v_{1,3}^2$] at (\leftx + 5*\dx, \confy) (conf132) {};
        \node[circle, draw, fill=black, label=below:$v_3^2$] at (\leftx + 5*\dx, \freey) (free32) {};

        \draw[c2, line width=\edgewidth pt] (free22) -- (conf122);
        \draw[c1, line width=\edgewidth pt] (conf122) -- (conf132);
        \draw[c3, line width=\edgewidth pt] (conf132) -- (free32);

        \node[circle, draw, fill=black, label=below:$b_{2,3}^1$] at (\leftx + 6*\dx, \bridgey) (bridge231) {};
        \node[circle, draw, fill=black, label=below:$b_{2,3}^2$] at (\leftx + 7*\dx, \bridgey) (bridge232) {};

        \node[circle, draw, fill=black, label=below:$v_1^3$] at (\leftx + 8*\dx, \freey) (free13) {};
        \node[circle, draw, fill=black, label=above:$v_{1,2}^3$] at (\leftx + 9*\dx, \confy) (conf123) {};
        \node[circle, draw, fill=black, label=below:$v_2^3$] at (\leftx + 10*\dx, \freey) (free23) {};

        \draw[c1, line width=\edgewidth pt] (free13) -- (conf123);
        \draw[c2, line width=\edgewidth pt] (conf123) -- (free23);

        \node[circle, draw, fill=black, label=below:$b_{3,4}^1$] at (\leftx + 11*\dx, \bridgey) (bridge341) {};
        \node[circle, draw, fill=black, label=below:$b_{3,4}^2$] at (\leftx + 12*\dx, \bridgey) (bridge342) {};

        \node[circle, draw, fill=black, label=below:$v_3^e$] at (\leftx + 13*\dx, \sparey) (spare3e) {};
        \node[circle, draw, fill=black, label=below:$v_3^b$] at (\leftx + 14*\dx, \sparey) (spare3b) {};
        \node[circle, draw, fill=black, label=below:$v_3^a$] at (\leftx + 15*\dx, \sparey) (spare3a) {};
        \node[circle, draw, fill=black, label=below:$v_3^c$] at (\leftx + 16*\dx, \sparey) (spare3c) {};
        \node[circle, draw, fill=black, label=below:$v_3^d$]     at (\leftx + 17*\dx, \sparey) (spare3d) {};

        \draw[c3prime, line width=\edgewidth pt] (spare3e) -- (spare3b);
        \draw[c3, line width=\edgewidth pt] (spare3b) -- (spare3a);
        \draw[c3prime, line width=\edgewidth pt] (spare3a) -- (spare3c);
        \draw[c3prime, line width=\edgewidth pt] (spare3a) to[bend left=50] (spare3d);

        \draw[b1, line width=\edgewidth pt] (free21) -- (bridge121);
        \draw[b1, line width=\edgewidth pt] (bridge121) -- (bridge122);
        \draw[b1, line width=\edgewidth pt] (bridge122) -- (free22);

        \draw[b2, line width=\edgewidth pt] (free32) -- (bridge231);
        \draw[b2, line width=\edgewidth pt] (bridge231) -- (bridge232);
        \draw[b2, line width=\edgewidth pt] (bridge232) -- (free13);

        \draw[b3, line width=\edgewidth pt] (free23) -- (bridge341);
        \draw[b3, line width=\edgewidth pt] (bridge341) -- (bridge342);
        \draw[b3, line width=\edgewidth pt] (bridge342) -- (spare3e);

        \node[circle, label=right:{Color Key:}] at (\keytitlex, \keyy) () {};

        \node[circle] at (\keyx, \keyy + \keydy) (c1left) {};
        \node[circle, label=right:$c_1$] at (\keyx + \keydx, \keyy + \keydy) (c1right) {};
        \draw[c1, line width=\edgewidth pt] (c1left) -- (c1right);

        \node[circle] at (\keyx, \keyy + 2*\keydy) (c1left) {};
        \node[circle, label=right:$c_2$] at (\keyx + \keydx, \keyy + 2*\keydy) (c1right) {};
        \draw[c2, line width=\edgewidth pt] (c1left) -- (c1right);

        \node[circle] at (\keyx, \keyy + 3*\keydy) (c1left) {};
        \node[circle, label=right:$c_3$] at (\keyx + \keydx, \keyy + 3*\keydy) (c1right) {};
        \draw[c3, line width=\edgewidth pt] (c1left) -- (c1right);

        \node[circle] at (\keyx, \keyy + 4*\keydy) (c1left) {};
        \node[circle, label=right:$c_3'$] at (\keyx + \keydx, \keyy + 4*\keydy) (c1right) {};
        \draw[c3prime, line width=\edgewidth pt] (c1left) -- (c1right);

        \node[circle] at (\keyx, \keyy + 5*\keydy) (c1left) {};
        \node[circle, label=right:$c_{1,2}^b$] at (\keyx + \keydx, \keyy + 5*\keydy) (c1right) {};
        \draw[b1, line width=\edgewidth pt] (c1left) -- (c1right);

        \node[circle] at (\keyx, \keyy + 6*\keydy) (c1left) {};
        \node[circle, label=right:$c_{2,3}^b$] at (\keyx + \keydx, \keyy + 6*\keydy) (c1right) {};
        \draw[b2, line width=\edgewidth pt] (c1left) -- (c1right);

        \node[circle] at (\keyx, \keyy + 7*\keydy) (c1left) {};
        \node[circle, label=right:$c_{3,4}^b$] at (\keyx + \keydx, \keyy + 7*\keydy) (c1right) {};
        \draw[b3, line width=\edgewidth pt] (c1left) -- (c1right);
    \end{tikzpicture}
    \caption{\label{fig:hardness} The construction given by~\Cref{thm:cfminecc-NPhard} for the CNF formula on three clauses $C_1 = (x_1 \lor x_2 \lor x_3)$, $C_2 = (x_1 \lor \neg x_2 \lor \neg x_3)$, and $C_3 = (\neg x_1 \lor \neg x_2)$. 
    Notice that the vertices are partitioned visually into four horizontal layers. The top layer contains \emph{conflict} vertices, the second from top contains \emph{free} vertices, the second from bottom \emph{bridge} vertices, and the bottom \emph{spare} vertices.
    Every edge containing a conflict vertex is a \emph{conflict} edge, every edge containing a bridge vertex is a \emph{bridge} edge, and all other edges are \emph{spare} edges.
    The colors $c_1, c_2$ and $c_3$ correspond to the clauses $C_1, C_2$, and $C_3$. The color $c_3'$ is the \emph{spare} color associated with $C_3$. The remaining colors are \emph{bridge} colors. We refer to the proof of~\Cref{thm:cfminecc-NPhard} for a formal description of the construction and the accompanying analysis.
    }

\end{figure}
}
\begin{proof}
    We reduce from \boolSAT{}, for which the input is a formula written in conjunctive normal form, consisting of $m$ clauses $C_1, C_2, \ldots C_m$ over $n$ variables $x_1, x_2, \ldots x_n$.
    For each variable $x_i$, we write $x_i$ and $\neg x_i$ for the corresponding positive and negative literals.
    We make the standard assumptions that no clause contains both literals of any variable, and that every literal appears at least once.
    We also assume that each clause has size $2$ or $3$, and that each variable appears in at most three clauses (implying that each literal appears in at most two clauses); hardness is retained under these assumptions~\cite{tovey1984simplified}.
    Given an instance of this problem, we construct an instance $(G = (V, E), \tau = 2)$ of \cfminECC{}. See~\Cref{fig:hardness} for a visual aid.

    For each clause $C_j$, we create a unique color $c_j$. Also, if $C_j$ has size two, we create a second unique color $c_j'$, and five vertices $v_j^a, v_j^b, v_j^c, v_j^d,$ and $v_j^e$.
    We call these the \emph{spare vertices} associated with $C_j$.
    The reason for creating these vertices and the second color $c_j'$ will become apparent later.
    Next, for each variable $x_i$, we create a vertex $v_{j_1, j_2}^i$ for each (ordered) pair of clauses $C_{j_1}, C_{j_2}$ with $C_{j_1}$ containing the positive literal $x_i$ and $C_{j_2}$ containing the negative literal $\neg x_i$.
    We call $v_{j_1, j_2}^i$ a \emph{conflict vertex} for the variable $x_i$ and the variable-clause pairs $(x_i, C_{j_1})$ and $(x_i, C_{j_2})$. Observe that each variable has either one or two associated conflict vertices, as does each variable-clause pair.
    For each variable-clause pair $(x_i, C_j)$ with a single associated conflict vertex, we create an additional vertex $v_j^i$ and call this the \emph{free vertex} for the variable-clause pair $(x_i, C_j)$.

    Now we create edges.
    For each clause $C_j$, we begin by creating one edge $e_j^i$ of color $c_j$ for each variable $x_i$ contained in $C_j$. This edge contains either the two $(x_i, C_j)$ conflict vertices, or the single $(x_i, C_j)$ conflict vertex and the $(x_i, C_j)$ free vertex.
    We call the edge $e_j^i$ a \emph{conflict edge} associated both with clause $C_j$, and with the literal of $x_i$ which appears in $C_j$.
    Observe that every free vertex has degree one, and every conflict vertex has degree two.
    For clauses $C_j$ of size $3$, there are exactly three edges of color $c_j$.
    For clauses $C_j$ of size $2$, we have thus far created two edges of color $c_j$.
    For each such clause, we use the associated spare vertices to create a third edge $\{v_j^a, v_j^b\}$ of color $c_j$, which we call the \emph{spare edge} associated with $C_j$, and three edges $\{v_j^a, v_j^c\}$, $\{v_j^a, v_j^d\}$, and $\{v_j^b, v_j^e\}$, each of color $c_j'$.
    Observe that our constructed graph still has maximum degree three, and that there are now exactly $3$ edges of every color.
    Finally, we set $\tau = 2$.

    Since free vertices have degree one, they do not participate in cycles.
    It is also clear from the construction that spare vertices do not participate in cycles.
    Thus, any cycle contains only conflict vertices.
    Note that if two conflict vertices share an edge, then they are associated with the same variable-clause pair.
    It follows that every vertex in a cycle is associated with some single variable-clause pair.
    However, a variable-clause pair has at most two associated conflict vertices, so we have constructed a forest.

    Now we will add some gadgets to turn our forest into a tree. Suppose that the graph we have constructed so far has $q$ connected components.
    Impose an arbitrary order on these components, and label them $G_1, G_2, \ldots, G_q$.
    Observe that every connected component contains either only spare vertices or two free vertices.
    In both cases, it is possible to add two edges each with one (distinct) endpoint in $G_i$ without raising the maximum degree above three.
    That the endpoints are distinct will be important when we analyze the cutwidth.
    For each consecutive pair $G_i, G_{i+1}$ of connected components, we add two \emph{bridge vertices} $b_{i, i+1}^1, b_{i, i+1}^2$ and a \emph{bridge color} $c_{i, i+1}^b$.
    We add \emph{bridge edges} (chosen so as not to violate our maximum degree constraint) from $G_i$ to $b_{i,i+1}^1$, from $b_{i, i+1}^1$ to $b_{i, i+1}^2$, and from $b_{i, i+1}^2$ to $G_{i+1}$.
    We color each of these three edges with $c_{i, i+1}^b$. Observe that the graph is now connected, but it is still acyclic since there is exactly
    one path between any pair of vertices which were in different connected components before the addition of our bridge gadgets.

    We now claim that the constructed graph $G$ has cutwidth~$2$.
    To this end, we will construct a ordering $\sigma\colon V \rightarrow \mathbb{N}$ of the vertices of $G$, where $\sigma$ is injective, $\sigma(v) = 1$ indicates that $v$ is the first vertex in the ordering, $\sigma(u) = |V|$ indicates that $u$ is the last vertex in the ordering, and $\sigma(v) < \sigma(u)$ if and only if $v$ precedes $u$ in the ordering.
    We will show that for every $i$, there exist at most two edges $uv \in E$ with the property that $\sigma(u) \leq i$ and $\sigma(v) > i$. We refer to this number of edges as the \emph{width} of the cut between vertices $i$ and $i+1$ in $\sigma$.
    We begin by guaranteeing that, for every $i$, if $u$ is the last vertex of $G_i$ in $\sigma$, $v$ is the first vertex of $G_{i+1}$ in $\sigma$, and $b_{i, i+1}^1, b_{i, i+1}^2$ are the associated bridge vertices, then $\sigma(u) = \sigma(b_{i, i+1}^1) - 1 = \sigma(b_{i, i+1}^2) - 2 = \sigma(v) - 3$.
    A consequence is that for every $i$, if $u$, $v$ are vertices of $G_i$ then every vertex $w$ with the property that $\sigma(u) < \sigma(w) < \sigma(v)$ is also a vertex of $G_i$.
    Observe that every cut between two bridge vertices has width $1$, as does every cut between a bridge vertex and a non-bridge vertex.
    We therefore need only consider cuts between pairs of vertices in the same $G_i$.
    If $G_i$ contains a conflict vertex, then $G_i$ is either a $P_3$ or a $P_4$.
    The former case arises when a variable appears in exactly two clauses, so there is one associated conflict vertex adjacent to two free vertices.
    The latter case arises when a variable appears in exactly three clauses, so there are two associated conflict vertices.
    These are adjacent.
    Additionally, each conflict vertex is adjacent to a distinct free vertex, so we have a $P_4$.
    Both the $P_3$ and $P_4$ have cutwidth $1$, as evidenced by ordering the vertices as they appear along the path.
    Moreover, we may assume that the relevant bridge edges are incident on the appropriate endpoints of the path, and so any cut between vertices of $G_i$ has width~$1$.
    Otherwise, $G_i$ contains no conflict vertices.
    In this case, $G_i$ contains only the five spare vertices $v_j^a, v_j^b, v_j^c, v_j^d$, and $v_j^e$ associated with some clause $C_j$, and the four edges are $v_j^av_j^b, v_j^av_j^c, v_j^av_j^d$, and $v_j^bv_j^e$.
    It is simple to check that this construction has cutwidth $2$, as evidenced by the ordering $\sigma(v_j^e) < \sigma(v_j^b) < \sigma(v_j^a) < \sigma(v_j^c) < \sigma(v_j^d)$.
    Once again, we may assume that the relevant bridge edges are incident on $v_j^e$ and $v_j^d$.
    Moreover, we can assume that at least one such $G_i$ exists, since \boolSAT{} instances in which every clause has size three and every variable appears at most three times are polynomial-time solvable~\cite{tovey1984simplified}.
    Thus, $G$ has cutwidth $2$.

    It remains to show that the reduction is correct. For the first direction, assume that there exists an assignment $\phi$ of boolean values to the variables $x_1, x_2, \ldots, x_n$ which satisfies every clause. We say that the assignment $\phi$ \emph{agrees} with a variable-clause pair $(x_i, C_j)$ if $C_j$ contains the positive literal $x_i$ and $\phi(x_i) = \textsf{True}$ or if $C_j$ contains the negative literal $\neg x_i$ and $\phi(x_i) = \textsf{False}$. We color the vertices of our constructed graph as follows.
    To every free vertex associated with clause $C_j$ we assign color $c_j$.
    The free vertex was only in edges of color $c_j$, so this results in zero unsatisfied edges.
    Next, to each bridge vertex we assign the associated bridge color.
    The bridge vertices were only in edges of this color, so once again this results in zero unsatisfied edges.
    Moreover, we have now guaranteed that at least one bridge edge of every bridge color is satisfied, meaning that at most two can be dissatisfied.
    Henceforth, we will not consider the remaining bridge edges.
    Next, for each clause $C_j$ of size $2$, we assign color $c_j'$ to every spare vertex associated with $C_j$. This results in exactly one unsatisfied edge of color $c_j$, and ensures that every edge of color $c_j'$ is satisfied. Finally, for each conflict vertex $v_{j_1, j_2}^i$, if $\phi$ agrees with $(x_i, C_{j_1})$ we assign color $c_{j_1}$ to $v_{j_1, j_2}^i$, and otherwise we assign color $c_{j_2}$.
    Observe that this coloring satisfies an edge $e_j^i$ if and only if $\phi$ agrees with $(x_i, C_j)$.
    Moreover, because $\phi$ is satisfying,
    every clause $C_j$ contains at least one variable $x_i$ such that $\phi$ agrees with $(x_i, C_j)$.
    Hence, at least one edge of every color is satisfied. Because there are exactly three edges of every color, there are at most two unsatisfied edges of any color.

    For the other direction, assume that we have a coloring which leaves at most two edges of any color unsatisfied. We will create a satisfying assignment $\phi$.
    For each variable $x_i$, we set $\phi(x_i) = \textsf{True}$ if any conflict edge associated with the positive literal $x_i$ is satisfied, and $\phi(x_i) = \textsf{False}$ otherwise.
    We now show that $\phi$ is satisfying. Consider any clause $C_j$. There are exactly three edges with color $c_j$, and at least one of them is satisfied.
    We may assume that the spare edge associated with $C_j$ (if such an edge exists) is unsatisfied, since satisfying this edge would require three unsatisfied edges of color $c_j'$. Thus,
    at least one conflict edge associated with $C_j$ is satisfied. Let $x_i$ be the corresponding variable, so the satisfied conflict edge is $e_j^i$.
    If $C_j$ contains the positive literal $x_i$, then $\phi(x_i) = \textsf{True}$ so $C_j$ is satisfied by $\phi$.
    Otherwise $C_j$ contains the negative literal $\neg x_i$.
    In this case, we observe that every conflict edge associated with the positive literal $x_i$ intersects with $e_j^i$ at a conflict vertex, and none of these edges has color $c_j$ since
    $C_j$ does not contain both literals. Hence, the satisfaction of $e_j^i$ implies that every conflict edge associated with the positive literal $x_i$ is unsatisfied.
    It follows that $\phi(x_i) = \textsf{False}$, meaning $C_j$ is satisfied by $\phi$.

    To show the claim for \cfmaxECC{}, we repeat the same construction, except that we omit all spare colors, vertices, and edges. The effect is that the constructed graph is a path.
    The proof of correctness is conceptually unchanged. Given a satisfying assignment $\phi$ we color vertices in the same way as before, and given a vertex coloring which satisfies
    at least one edge of every color, it remains the case that at least one conflict edge associated with every clause must be satisfied.
    The remaining analysis is similar.
\end{proof}


%




\section{Full proof of~\texorpdfstring{\Cref{thm:pc-ECC-FPT}}{}}
\label{app:pc}


\pceccFPT*
\begin{proof}
    We give a branching algorithm which is essentially identical to that of~\Cref{thm:cfminecc-FPT}.
    For completeness, we repeat the details.
    Given an instance $(H = (V, E), t, b)$ of \pcECC{}, a \emph{conflict} is a triple $(v, e_1, e_2)$ consisting of a single vertex $v$ and a pair of distinctly colored hyperedges $e_1, e_2$ which both contain $v$.
    If $H$ contains no conflicts, then it is possible to satisfy every edge.
    Otherwise, we identify a conflict in $O(r|E|)$ time by scanning the set of hyperedges incident on each node.
    Once a conflict $(v, e_1, e_2)$ has been found, we branch on the two possible ways to resolve this conflict: deleting $e_1$ or deleting $e_2$.
    Here, deleting a hyperedge has the same effect as ``marking'' it as unsatisfied and no longer considering it for the duration of the algorithm.
    We note that it is simple to check in constant time whether a possible branch violates the constraints given by $t$ or $b$; these branches can be pruned.
    Because each branch increases the number of unsatisfied hyperedges by 1 and WLOG $t > b$, the search tree has depth at most $t$.
    Thus, the algorithm runs in time $O(2^{t}r|E|)$.
\end{proof}

\end{document}